\documentclass[sn-basic]{sn-jnl}

\usepackage{natbib}
\usepackage{amssymb, amsmath}
\usepackage{mathtools} 

\DeclareMathOperator*{\argmax}{argmax}
\DeclareMathOperator*{\argmin}{argmin}
\DeclareMathOperator*{\supp}{supp}
\DeclareMathOperator*{\pers}{pers}
\DeclareMathOperator*{\length}{length}

\renewcommand{\R}{\mathbb{R}}
\renewcommand{\Z}{\mathbb{Z}}
\newcommand{\N}{\mathbb{N}}
\newcommand{\Se}{\mathbb{S}}
\newcommand{\bottleneck}{d_B}
\newcommand{\indicator}{1}
\newcommand{\diff}{\mathrm{d}}
\newcommand{\geometric}[1]{\overline{#1}}
\newcommand{\restr}[2]{{
		\left.\kern-\nulldelimiterspace 
		#1 
		\right\vert_{#2} 
}}
\newcommand{\freq}{\omega}
\newcommand{\Partition}{\mathcal{A}}
\newcommand{\partition}{A}

\jyear{2022}%

\theoremstyle{thmstyleone}%
\newtheorem{theorem}{Theorem}
\newtheorem{proposition}[theorem]{Proposition}%
\newtheorem{lemma}[theorem]{Lemma}

\theoremstyle{thmstyletwo}%
\newtheorem{remark}{Remark}%

\theoremstyle{thmstylethree}%
\newtheorem{problem}{Problem}%
\newtheorem{definition}{Definition}%

\begin{document}

\title[Topological phase estimation]{Topological phase estimation method for reparameterized periodic functions}

\author[1,2]{\fnm{Thomas} \sur{Bonis}}\email{thomas.bonis@univ-eiffel.fr}
\author[3]{\fnm{Fr\'ed\'eric} \sur{Chazal}}\email{frederic.chazal@inria.fr}
\author[4]{\fnm{Bertrand} \sur{Michel}}\email{bertrand.michel@ec-nantes.fr}
\author*[3]{\fnm{Wojciech} \sur{Reise}}\email{wojciech.reise@inria.fr}

\affil[1]{\orgdiv{LAMA}, \orgname{Universit\'e Gustave Eiffel}, \orgaddress{\street{5 boulevard Descartes}, \city{Champs--sur--Marne}, \postcode{77420},
		\country{France}}}
\affil[2]{\orgname{Sysnav}, \orgaddress{\street{72 Rue Emile Loubet}, \city{Vernon}, \postcode{27200}, \country{France}}}
\affil[3]{\orgname{Inria Saclay}, \orgaddress{\street{1 rue Honor\'e d'Estienne d'Orves}, \city{Palaiseau}, \postcode{91120}, \country{France}}}
\affil[4]{\orgname{\'Ecole Centrale de Nantes}, \orgaddress{\street{1 Rue de la No\"e}, \city{Nantes}, \postcode{44300}, \country{France}}}

\abstract{
	We consider a signal composed of several periods of a periodic function, of which we observe a noisy reparametrisation. The phase estimation problem consists of finding that reparametrisation, and, in particular, the number of observed periods. Existing methods are well--suited to the setting where the periodic function is known, or at least, simple. We consider the case when it is unknown and we propose an estimation method based on the shape of the signal. We use the persistent homology of sublevel sets of the signal to capture the temporal structure of its local extrema. We infer the number of periods in the signal by counting points in the persistence diagram and their multiplicities. Using the estimated number of periods, we construct an estimator of the reparametrisation. It is based on counting the number of sufficiently prominent local minima in the signal. This work is motivated by a vehicle positioning problem, on which we evaluated the proposed method.
}




\maketitle

\section{Introduction}
\label{sec:intro}
In this work, we study an inference problem on a re-parametrisation of a function. Specifically, consider
\begin{equation}
\label{eq:noise_model}
\begin{aligned}
S: [0,1]&\rightarrow\R \\
t &\mapsto (f\circ \gamma)(t) + W_t,
\end{aligned}
\end{equation}
where $(W_t)_{t\in [0,T]}$ is a stochastic process, $f$ is a one--periodic, continuous function and $\gamma:[0,1]\rightarrow [0,N]$ is an increasing bijection for some $N\in\N^*$.
\begin{figure}
	\centering
	\includegraphics[width=0.3\textwidth]{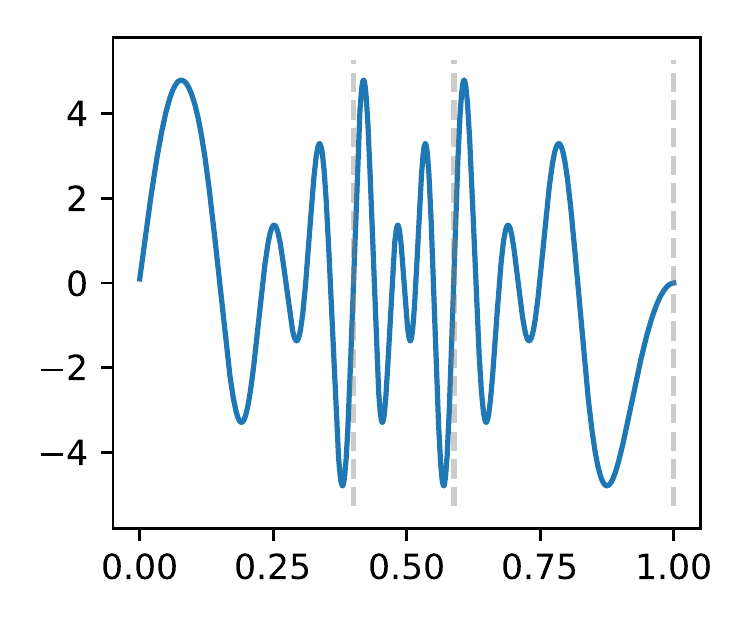}
	\includegraphics[width=0.3\textwidth]{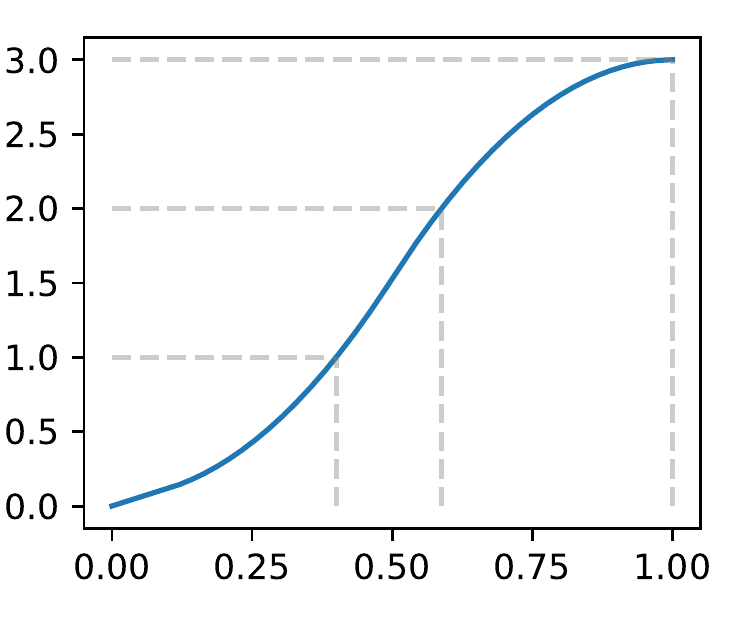}
	\caption{On the left, an example of a noise--free observed function $f\circ\gamma$, where a one--periodic $f$ has been composed with $\gamma$ depicted on the right. On both graphs, $t$ such that $\gamma(t)\in \N$ are marked by the gray dashed lines.}
	\label{fig:model_illustration}
\end{figure}
We consider the setting where $f$, $\gamma$ and $N$ are unknown. Figure~\ref{fig:model_illustration} illustrates the model with an example of $f$ featuring eight local extrema per period, a non--linear $\gamma$ and $W_t=0$. We are interested in two inference problems on $\gamma$, motivated by industrial applications to positioning systems.
\begin{problem}
	\label{prob:estimate_N}
	Given $S$, estimate $N$.
\end{problem}
\begin{problem}[Odometry]
	\label{prob:odometry}
	Given $S$ and $N$, find $(t_i)_{i=1}^N$, such that
	\begin{equation*}
	\label{eq:odometry_condition}
	\gamma(t_i)- \gamma(t_{i-1})=1, \qquad \forall\ i=2,\ldots, N.
	\end{equation*}
	Such a sequence will be called an odometric sequence.
\end{problem}
Problem~\ref{prob:estimate_N} consists of estimating how many periods of $f$ are present in an observed signal $S$. Problem~\ref{prob:odometry} is a segmentation task, where the number of segments is given. This segmentation, called an odometric sequence, thus contains finer information: if we know that $(t_i)_{i=1}^N$ is an odometric sequence for $S$, then $\restr{S}{[t_1, t_N]}$ contains $N-1$ periods. We therefore include the number of periods $N$ as part of the input data to Problem~\ref{prob:odometry}.

\subsection{Related literature}
\label{sec:lit_review}
Inverse problems similar to Problems~\ref{prob:estimate_N} and~\ref{prob:odometry} have been studied in different contexts.
In signal processing, the instantaneous phase estimation~\citep[Chapter 10]{boashash_time-frequency_2015} is concerned with recovering $t\mapsto\gamma(t)$ when $f$ is known to belong to a family of functions of simple form, for example $f(x) = a\sin(2\pi(x-\phi))$. Such models were motivated by applications in communications, where the carrier and signal waves have a single sinusoidal component. The work has led to two types of methods. Spectral methods consist of using a frequency representation of $S$ to recover the phase: the exact solution is given by $\gamma(t) = \arctan(H(S)(t)/S(t))$, where $H$ is the Hilbert transform~\citep{boashash_algorithms_1990}. For more complex signals, one can analyze time frequency representations of the signal, using image processing techniques like peak detection and component linking~\citep{rankine_if_2007}. This procedure is even extended to sums of independently-parametrized functions~\citep{khan_multi-component_2016, hussain_adaptive_2002}. 
These methods are aimed at estimating $\gamma'$ and it is shown possible under hypothesis of separation of different components in the frequency spectrum.
Another type of method is concerned with geometric properties of the signal. The simplest example, called the \emph{zero crossings} method relies on counting the number of times that the graph of $S$ crosses zero~\citep{boashash_algorithms_1990}. This principle is generalized to other forms via intrinsic mode functions~\citep{huang_empirical_1998}. With the empirical mode decomposition algorithm~\citep{huang_empirical_1998}, $f$ can be decomposed into more elementary functions, connecting local maxima and minima using cubic splines. This algorithm decomposes the signal into multiple, independent, simple components.
On one hand, the setting considered in our work is easier than that for the spectral methods, because~\eqref{eq:noise_model} implies two properties of the signal. First, it has constant amplitude. Second, if $f$ is decomposed into multiple components, those components share a reparametrisation. On the other hand, nothing other than that is known about $f$ and the observations are corrupted by additive noise $W$.

Topological methods for signal processing become increasingly widespread~\citep{robinson_topological_2014}.
A widely known application to the study of time series is detecting the recurrent structure in dynamical systems~\citep{perea_persistent_2016}. This is achieved through the lens of the homology of the time delay embedding of the time series.
Beyond topology, the understanding of the shape of that embedding extends to geometric and spectral decomposition properties~\citep{gakhar_kunneth_2019}. In particular, the geometric and topological features depend not only on $f$, but also on the frequencies in the signal.
The Hilbert transform embedding method has been proposed in~\cite{kennedy_novel_2018}. It is a different embedding method inspired by the analytic signal transformation found in the signal processing literature. It is shown to maximize the topological signal of a reparametrized, sinusoidal function.
In our context, the time delay embedding technique is not adapted, because the frequencies of the signal vary in a range we do not control. The signature of a path~\citep{hambly_uniqueness_2010} is a bounded sequence, which is invariant to any reparametrisation. However, we are not aware of any relation between the signature of a single period and that of multiple periods of a function. 

An alternative line of work at the intersection of signal processing and topological methods consists of studying a one dimensional signal directly. Similarly to the zero crossings or empirical mode decomposition from signal processing, it consists of looking at the shape of the signal. More precisely, it is based on the homology of the sublevel sets of the signal. This approach led to a denoising method~\citep{plonka_relation_2016} and also to confidence intervals on a topological descriptor~\citep{myers_separating_2020}. In~\cite{khasawneh_topological_2018}, the authors propose a period counting method for binary functions $f$. Similar to the zero crossings, the method counts the number of changes in a binary signal, but discards those which last for only a short period of time. As shown experimentally by the authors, this makes it more robust than the spectral methods. That method solves both Problems~\ref{prob:estimate_N} and~\ref{prob:odometry}, but in the simple case where $f$ is a binary signal.

As an example, Figure~\ref{fig:introduction_f} shows three functions $f$. The two on the left come from the model for the zero crossings method and the topological method from \cite{khasawneh_topological_2018} respectively. Neither of the two methods was designed for continuous signals with many local extrema per period, as presented on the right, which we attempt to study.

\begin{figure}
	\centering
	\includegraphics[width=0.7\textwidth]{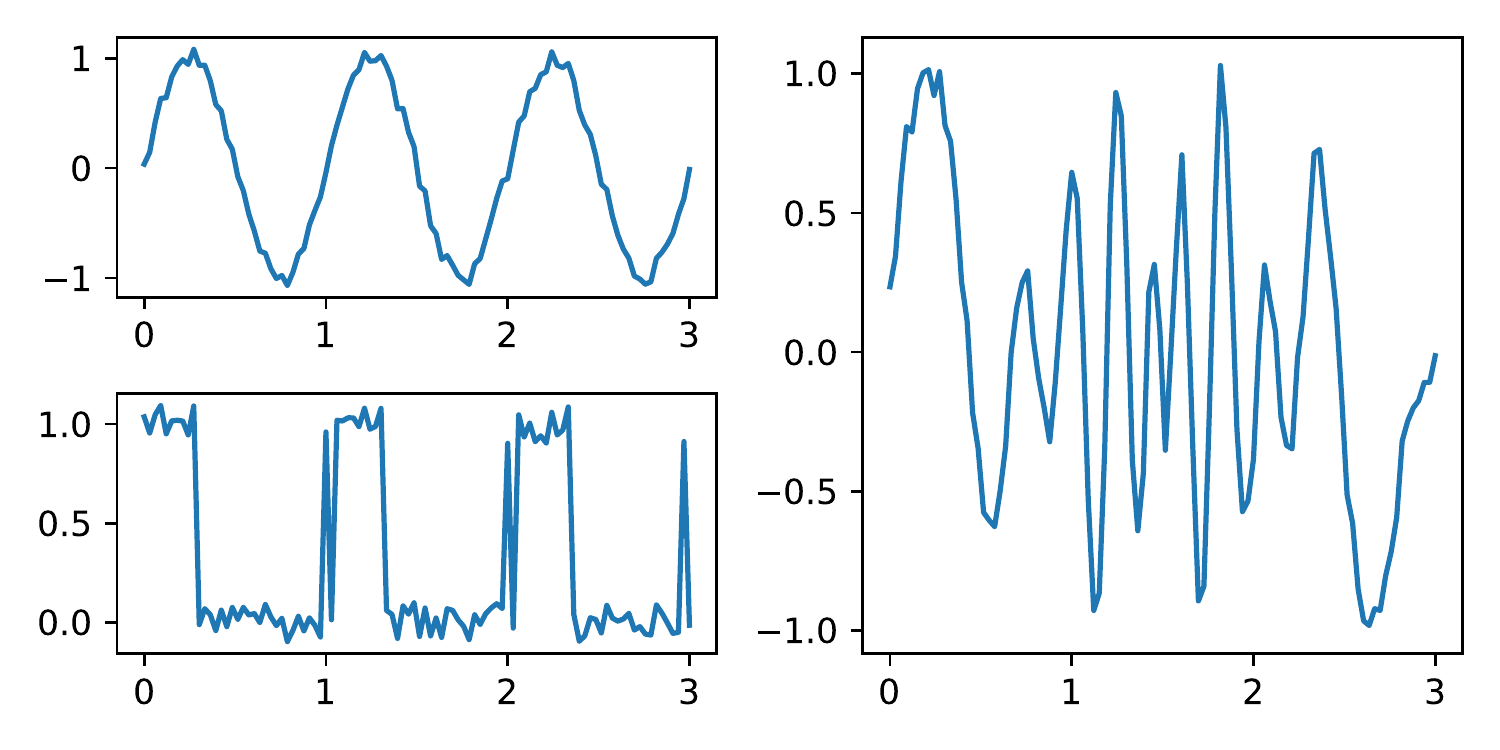}
	\caption{
		On the left, the graphs of signals from models which the zero crossings and the topological method~\cite{khasawneh_topological_2018} respectively, were designed for. On top, a sinusoidal function with one pair of extrema per period. On the bottom, a noisy binary signal.
		On the right, a signal following the model \eqref{eq:noise_model} studied in this work.}
	\label{fig:introduction_f}
\end{figure}

\subsection{Practical motivation}
This work is motivated by a positioning problem. Given a collection of samples from $S$ measuring the magnetic field inside a vehicle at position $\gamma(t)$, the goal is to retrieve information about the distance, $\gamma$, traveled by that vehicle.
Magnetic field measurements have been used to robustify and enhance the precision of position estimates~\citep{cruz_magland_2020, wei_vehicle_2017}. However, in these approaches, the magnetic fingerprint of the environment is either assumed known or constructed first. Then, at inference time, the observed magnetic signal is compared with the gathered ``reference signals".

The methodology proposed in this work relies on a different principle, which we credit to preliminary experiments conducted in~\cite{bristeau_techniques_2012}. When a magnetic sensor is placed near a wheel of the vehicle moving on a straight line, the rotating wheel generates perturbations in the measurements of the magnetic field. If we assume that Earth's magnetic field is constant, these perturbations are periodic in the rotation of the wheel. Solving Problem~\ref{prob:odometry} is therefore akin to counting number of revolutions of the wheel, up to a multiplicative constant, and it gives us an approximation of $\gamma$. Given the odometric sequence $(t_n)_{n=1}^N$, we define $\hat{\gamma}$ as an estimator of $\gamma$
\begin{equation*}
\hat{\gamma}(t) = \sum_{n=1}^N \indicator_{t_n\leq t}.
\end{equation*}
The dependency of the signal on the rotation of the wheel is not understood via physical modeling of the magnetic field, so little is known about $f$. In particular, car wheels often feature symmetric structures, so it is possible that $f$ has many extrema per period. Little is understood about other sources of magnetic perturbations and its nature can vary. Electronic devices in the car such as blinkers or screen wipers, produce their own magnetic field, which is also periodic (in time, not in the distance traveled by the vehicle). Infrastructure like high tension lines or metal structures are another source of magnetic noise. This justifies the introduction of model~\eqref{eq:noise_model} with additive noise, which is smaller (or of the order of) than the amplitude of the periodic functions' oscillations.

\subsection{Contributions}
In this work, we propose to use topological methods for Problems~\ref{prob:estimate_N} and~\ref{prob:odometry} on $\gamma$. We work under genericity assumptions on $f$ and an upper bounded noise $W$.
In Proposition~\ref{prop:periodic_multiplicity}, we express the persistence diagram of $N$ periods of $f$ using the diagram of a period of $f$. Using that result, we propose an estimator of $\gamma(1)$ in Section~\ref{sec:method_multiplicity_points}. We show this estimator solves Problem~\ref{prob:estimate_N} under genericity assumptions on $f$ and on the amplitude of the noise. We propose a robust alternative suitable in practice in Section~\ref{sec:pratical_adaptations} and test its performance in Section~\ref{sec:numerical_results}.
Problem~\ref{prob:odometry} is tackled in Section~\ref{sec:method_odometry}, where we introduce a method to calculate an odometric sequence from the data contained in the persistence diagram.
We test the estimator for Problem~\ref{prob:estimate_N} on synthetic signals and we show the applicability of the odometric sequence in practice in Section~\ref{sec:magnetic_odometry}.

\section{Topological descriptors of real-valued functions}
\label{sec:background_persistence}
We introduce the necessary background in persistent homology of dimension zero, applied to scalar functions. For a more exhaustive treatment of algebraic topology or persistent homology, we refer the interested reader to \citep{hatcher_algebraic_2002} and \citep{edelsbrunner_topological_2002} respectively. We end this section with Proposition~\ref{prop:periodic_multiplicity}, which characterizes the persistence diagram for periodic functions.

Let $f: X \rightarrow \R$ be a continuous function with a finite number of local extrema, defined on a compact and connected topological space $X$. Consider the collection of increasing spaces $(f^{-1}(\rbrack-\infty,r\rbrack))_{r\in\R}$, called the sublevel set filtration of $X$ induced by $f$, and the set of local minima $\mathcal{C}$ of $f$
\begin{equation*}
\mathcal{C} = \{c\in
X
\mid \exists V_c \text { a neighborhood of } c,\ f(c)\leq f(x),\, \forall x\in V_c \}.
\end{equation*}
The persistence diagram $D(f)$ of $f$ is a multi--set in $\R^2$, capturing information about the local extrema of $f$ and topological information about $X$. Starting from $r=-\infty$ and increasing $r$, each element $c\in\mathcal{C}\subset
X
$ appears for the first time in level $r=s(c)$ as an isolated point.  For this local minimum, we will add a point in the persistence diagram, with coordinates $(b,d)$, called the birth and death value respectively. The birth $b=f(c)$ is the value at which the component first appears. The death $d$ is the supremum of values $y$, such that $c$ realizes the minimum of $f$ on its connected component in $f^{-1}(\rbrack -\infty,\, y\rbrack)$. In addition, we will say that the point $(b,d)$ is associated to the local minimum $c$, or that $c$ is the birth point of $(b,d)$. This defines an additional function $\mathcal{C}\rightarrow D\subset \R^2$, associating to each local minimum a point in the diagram.
We note that since $X$ is connected and $f$ has a finite number of local extrema, for any $\alpha>M$, where $M$ is the largest local maximum value, $f^{-1}(\rbrack-\infty, \alpha\rbrack)=X$ has a single connected component. We are exclusively interested in functions defined on a bounded interval, that we can suppose $X=[0,1]$ without loss of generality. Inspired by the presentation of the persistence algorithm in~\cite{chazal_persistence-based_2013}, we write Algorithm~\ref{alg:persistence}, which summarizes how the persistence diagram and the function $\mathcal{C}\rightarrow D$ are calculated. In our simulations, we use the SimplexTree from the GUDHI library~\citep{gudhi} to calculate the persistence diagrams and manipulate the complex.

\begin{algorithm}
	\begin{algorithmic}[1]
		\Procedure{Persistence diagram}{f}
		\State $D  \leftarrow \{\}$
		\For{$c$ a local extremum of $f$, ordered by ascending value of $f$}
		\If{$c$ is a local min}
		\State Push $(f(c),*)$ to $D$ with key $c$.
		\Else \Comment{$c$ is a local max}
		\State $I \leftarrow$ connected component of $c$ in $f^{-1}(\rbrack -\infty, f(c) \rbrack)$.
		\State $c_1,\,c_2 \leftarrow \argmin_{x<c} f(I)$, $\argmin_{x>c} f(I)$.
		\State $i \leftarrow \argmax\{f(c_1), f(c_2)\}$
		\State Set the second coordinate of $D(c_i)$ to $f(c)$.
		\EndIf
		\EndFor
		\State Set the second coordinate of $D(\argmax(s))$ to $\max(f)$.
		
		\Return{$D$}
		\EndProcedure
	\end{algorithmic}
	\caption{Computing the annotated persistence diagram of dimension 0 of a sublevel set filtration of $f$}
	\label{alg:persistence}
\end{algorithm}
The persistence of a point $(b,d)$, $\pers(b,d) = d-b$ is often associated with the importance of the corresponding topological feature. The set of points with zero persistence is the diagonal $\Delta=\{(x,x)\mid x\in\R\}$. It is often considered part of the persistence diagram, with each point included with infinite multiplicity. This is justified at the level of the algebraic definition of persistence and motivated by stability considerations. For $\tau>0$, we will denote by $\Delta_\tau = \{(b,d)\mid \pers(b,d)\leq 2\tau\}$ the $\tau$ thickening of the diagonal, which is the set of points with persistence less than $2\tau$.

A matching between two diagrams $D,\,D'$ is a bijection $M: (D\cup\Delta)\rightarrow (D'\cup\Delta)$.
The cost of that matching is the maxium distance between two matched elements $\max_{a\in D\cup\Delta}\Vert a-M(a) \Vert_\infty$. The bottleneck distance \cite[VIII.2]{herbert_edelsbrunner_computational_2010} between two diagrams is defined as the infimum cost a matching $M$ over all possible matchings $M: (D\cup\Delta)\rightarrow (D'\cup\Delta)$
\begin{equation*}
\label{eq:bottleneck_distance}
\bottleneck(D, D') = \inf\limits_{M} \sup_{a\in D\cup\Delta}\Vert a-\Gamma(a) \Vert_\infty.
\end{equation*}
It is a classical result that the persistence diagram associated to a function is then stable in bottleneck distance~\citep{cohen-steiner_stability_2007, chazal_structure_2016}. Formally, for two diagrams $D(f), D(g)$ associated to functions $f,\,g$,
\begin{equation}
\label{eq:persistence_stability}
\bottleneck(D(f), D(g)) \leq \Vert f - g\Vert_\infty.
\end{equation}
We recall that $d_1(D(f), 0) = \sum_{p\in D(f)} \pers(p)/2$.

While persistence diagrams are often viewed as multi--sets, they can be interpreted as discrete measures. For the persistence diagram $D(f)$, we define the discrete measure $\mu_f$ associated to that diagram as
\begin{equation*}
\label{eq:diagram_as_measure}
\mu_f = \sum_{p\in D(f)} \delta_p,
\end{equation*}
where $\delta_p$ denotes the point measure of mass one at $p$. With this notations, the number of points in $A\subset \R^2$ is given by $\mu_f(A)$. We introduce the notation for the multiplicity of a point
\begin{equation*}
\langle p\rangle_f = \mu_f(\{p\}).
\end{equation*}
The bottleneck distance between measures can be defined for measures coming from persistence diagrams.

Let $f:\R\rightarrow\R$ be a continuous function, with period one and a finite number local extrema. We can see  $\restr{f}{[0,1]}$ the restriction of $f$ to $[0,1]$ as a function on the circle $\bar{f}:\Se^1\rightarrow \R$, via the canonical projection $\R\rightarrow \Se^1 \simeq \R/\Z, t\mapsto (\cos(2\pi t), \sin(2\pi t))$.
\begin{proposition}
	\label{prop:periodic_multiplicity}
	For $N\in\N$, we consider
	\begin{equation*}
	\begin{array}{cccc}
	f_N:&[0,1]&\rightarrow&\R\\
	&t&\mapsto&f(Nt).
	\end{array}
	\end{equation*} Then,
	\begin{equation*}
	\mu_{\bar{f_N}} = N\mu_{\bar{f}}.
	\end{equation*}
\end{proposition}
We refer the interested reader to Appendix~\ref{app:proof_periodic_multiplicity} for a proof. Proposition~\ref{prop:periodic_multiplicity} states that the function $N\mapsto \mu_{\bar{f_N}}$ is homogeneous. An analogue statement is not valid for $X$ being an interval, because of boundary effects. In the remaining of this work, we will heavily rely on Proposition~\ref{prop:periodic_multiplicity}. We will abuse notation, by identifying the persistence diagram of $f$ with the persistence diagram of $\bar{f}$.

\section{Estimation of $N$ in the noiseless setting}
\label{sec:model}
Let $f:\R\rightarrow\R$ be one--periodic and $\gamma:[0,1]\rightarrow [0,N]$ an increasing bijection for $N\in\N^*$, as above. In this section, we consider the noiseless setting $S=f\circ\gamma$.
We introduce an estimator of $N$
\begin{equation*}
N(f) = \gcd\{\langle p\rangle_f \mid p\in \supp(\mu_f)\}.
\end{equation*}
Recall that for a set $A\subset \N$, the greatest common divisor $\gcd(A)$ is the largest $k\in\N$ such that for all $a\in A$, $k$ divides $a$. For $A\subset \N^*$ non empty, $1\leq \gcd(A)<\infty$. We adopt the convention that $\gcd(\emptyset)=1$.
The proposed estimator satisfies a homogeneity property, in the sense detailed by Proposition~\ref{prop:N}.
\begin{proposition}
	\label{prop:N}
	Suppose that $\gamma(1)=N$. Then, 
	$$N(f\circ \gamma) = N\cdot N(\restr{f}{[0,1]}).$$
\end{proposition}
\begin{proof}
	Consider $f_N$ as defined in Proposition~\ref{prop:periodic_multiplicity}. The connected components of the sublevel sets
	$$(f\circ\gamma)^{-1}(\rbrack-\infty, a\rbrack) = \{t\mid f(\gamma(t))\leq a\} = \gamma^{-1}(\{x\mid f_N(x)\leq a\})$$
	are in bijection with the components of $f_N^{-1}(\rbrack-\infty, a\rbrack)$. Hence, $\mu_{f_N}=\mu_{f\circ\gamma}$ and $\supp(\mu_{f_N})=\supp(\mu_{f_1})$. By Proposition~\ref{prop:periodic_multiplicity}, we obtain $\mu_{f\circ\gamma} = \mu_{f_N} = N\mu_{f_1} = N\mu_{\restr{f}{[0,1]}}$, so that
	\begin{align*}
	N(f\circ \gamma)
	&= \gcd\{\mu_{f\circ\gamma}(\{p\}) \mid p\in \supp(\mu_{f_N})\}\\
	&= \gcd\{ N\mu_{\restr{f}{[0,1]}}(\{p\}) \mid p\in \supp(\mu_{f_1})\}\\
	&= N\cdot \gcd\{\langle p\rangle_{\mu_{\restr{f}{[0,1]}}} \mid p\in \supp(\mu_{f_N})\}\\
	&= N\cdot N(\restr{f}{[0,1]}).
	\end{align*}
\end{proof}
The estimator $N(f)$ is correct only up to a multiplicative constant $N(\restr{f}{[0,1]})$, so we need to introduce the notion of non--degenerate functions.

\begin{definition}
	\label{def:non_degenerate}
	We call a one--periodic function $f:\R\rightarrow\R$ non--degenerate if
	\begin{equation}
	\label{eq:hyp_gcd}
	N(\restr{f}{[0,1]})=1.
	\end{equation}
	A function $f:\R\rightarrow\R$ which does not satisfy \eqref{eq:hyp_gcd} will be called degenerate.
\end{definition}
The non--degeneracy of a function $f$ is a condition on the set of pairs of local extreme values and Proposition~\ref{prop:assumption_is_necessary} provides a justification for restricting our considerations to non--degenerate functions. Indeed, if $f$ is a degenerate function, then there exists a function $g$, which has the same persistence diagram as $f$, but for which Problem~\ref{prob:estimate_N} is not identifiable.

\begin{proposition}
	\label{prop:assumption_is_necessary}
	Let $f:\R\rightarrow\R$ be a one--periodic function. There exists a $1/N(\restr{f}{[0,1]})$--periodic function $g:\R\rightarrow \R$, such that
	$$\mu_{\restr{f}{[0,1]}}=N(\restr{f}{[0,1]})\mu_{\restr{g}{[0,1]}}.$$
\end{proposition}
\begin{figure}
	\centering
	\includegraphics[width=0.7\textwidth]{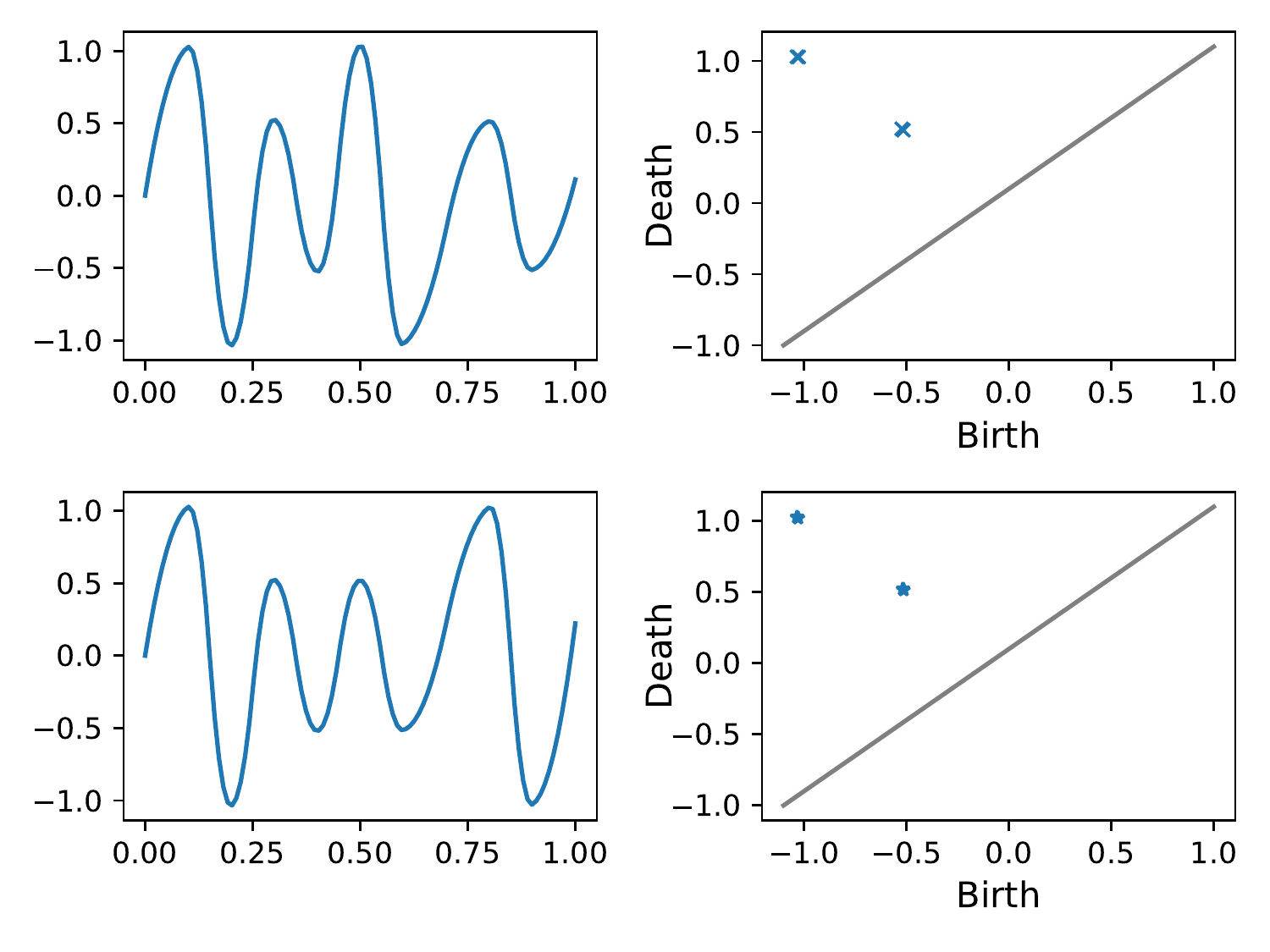}
	\caption{Graphs of two degenerate one--periodic functions and their persistence diagrams.}
	\label{fig:hyp_gcd_illustration}
\end{figure}
We discuss the identifiability of the problem and the degeneracy in Figure~\ref{fig:hyp_gcd_illustration}, which shows the graphs and persistence diagrams of two degenerate functions.
For $f$ in the top row, the problem of inference on $\gamma$ is not identifiable, since $\restr{f}{[0,1/2]}$ could be re--parametrized to have the same graph as $\restr{f}{[1/2,1]}$. In particular, the function $f$ is degenerate. The function $g$ is also degenerate, because it induces the same persistence diagram as $f$. However, its extremal values occur in a different order, so the same re-parametrization trick cannot be applied and the problem is identifiable. This situation exemplifies a limitation of our method and more generally, of making the inference using only the persistence diagram.

\begin{figure}
	\centering
	\includegraphics[width=0.34\textwidth]{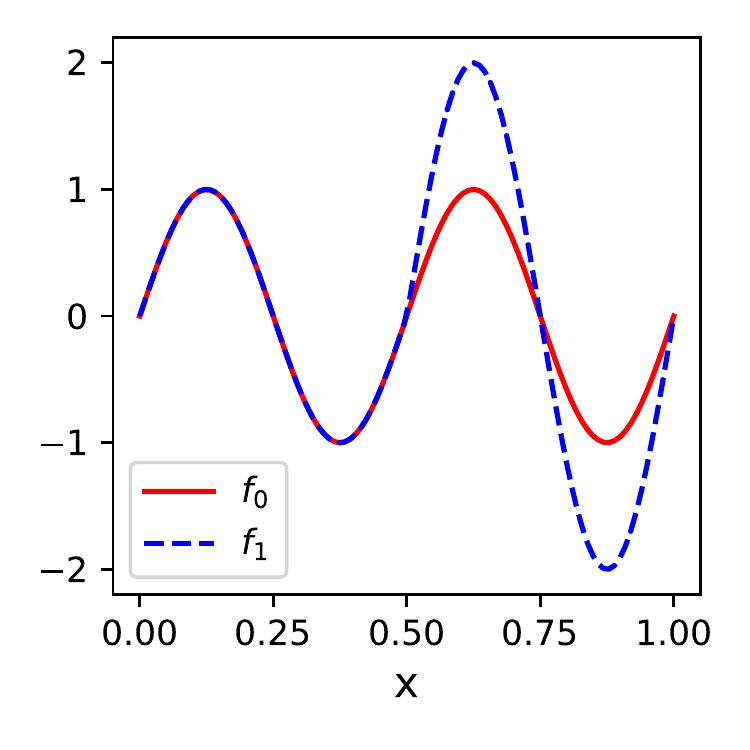}
	\includegraphics[width=0.34\textwidth]{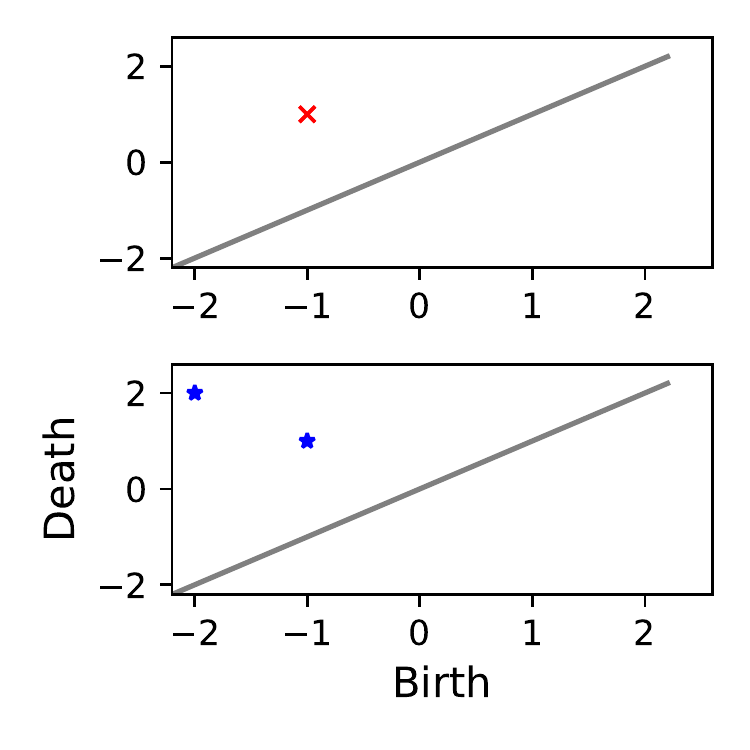}
	\caption{On the left, graphs of $\restr{f_0}{[0,1]}$ and $\restr{f_1}{[0,1]}$. On the right, their persistence diagrams.}
	\label{fig:delta_degenerate_example}
\end{figure}
Non--degeneracy~\eqref{eq:hyp_gcd} restricts the class of considered functions, but it does not quantify the difficulty of the problem of estimation of $N$.
Figure~\ref{fig:delta_degenerate_example} illustrates two functions from a family of functions $(f_r)_{r\geq0}$
\begin{equation*}
\label{eq:degenerate_periodic}
f_r(x) =
\begin{cases}
\sin(4\pi x), &\ x\in[0,1/2]\\
(1+r)\sin(4\pi x), &\ x \in \rbrack 1/2,1\rbrack
\end{cases},
\end{equation*}
and extended by 1-periodicity to $\R$.
Here, $f_0$ is half--periodic and degenerate, while $f_r,\,r>0$ is not. Indeed, the persistence measure $\mu_{f_r}$ is supported on the points $(-1-r, 1+r)$ and $(-1, 1)$, each with multiplicity 1.
 As $r\rightarrow0$, $f_r\xrightarrow{\Vert\cdot\Vert_\infty} f_0$ and, while $\bottleneck(\mu_{f_0}, \mu_{f_r})\rightarrow 0$, we have $N(\restr{f_r}{[0,1]}) = 1$ for all $r>0$. 
To quantify degeneracy, we introduce $\delta$ which measures the separation of points in a diagram $\mu$, from each other and from the diagonal
\begin{equation}
\label{cond:separation}
\delta = \min(\{d(p,q) \mid p, q\in \supp(\mu),\, p\neq q\} \cup \{d(p,\Delta) \mid p\in \supp(\mu)\}).
\end{equation}
The next proposition lower--bounds how far a non--degenerate function is from a degenerate function, in terms of its separation $\delta$.

\begin{proposition}
	\label{prop:distance_to_degenerate}
	Let $f:\R\rightarrow\R$ be a non--degenerate function and $\delta>0$ be the separation of $\mu_f$. For any degenerate $g: \R\rightarrow \R$ with $N(\restr{g}{[0,1]})=n\geq 2$,
	\begin{equation*}
	\bottleneck(\mu_{\restr{f}{[0,1]}}, \mu_{\restr{g}{[0,1]}}) \geq\frac{\delta}{2}.
	\end{equation*}
\end{proposition}
\begin{proof}
	For this proof, we consider the diagram as a multi--set. In the argument below, we need to make a distinction between a point $p$ in a diagram, and its geometric realization in $\R^2$, which we denote by $\geometric{p}$: there can be multiple points in $D(f)$ which have the same geometric realizations.
	Let $M: (D(\restr{f}{[0,1]})\cup\Delta) \rightarrow (D(\restr{g}{[0,1]})\cup\Delta)$ be a matching.
	
	As $N(\restr{g}{[0,1]}) = n$ does not divide $1= N(\restr{f}{[0,1]})$, there is $p_0\in\supp(D(\restr{f}{[0,1]}))$ such that
	\begin{equation}
	\label{eq:p_one_definition}
	n \text{ does not divide } \langle \geometric{p_0} \rangle_{\restr{f}{[0,1]}}.
	\end{equation}
	Let us denote by $P=\{p\in D(f)\mid \geometric{p}=\geometric{p_0}\}$ the set of points from the diagram, which have the same geometric realization as $p_0$ and $Q=M(P) = \{ M(p)\mid p\in P\}\subset \R^2$ the image of these points by the matching.
	
	Consider $Q'=\{q'\in D(\restr{g}{[0,1]})\mid \exists q\in Q,\, \geometric{q}=\geometric{q'}\}$ the subset of $D(\restr{g}{[0,1]})$ with geometric realizations in $Q$, and the set $P'=  M^{-1}(Q')\setminus P$.
	Let us show that $P'\neq \emptyset$. Since $ M(P)=Q\subset Q'$, if $P'=\emptyset$, then $P =  M^{-1}(Q')$. The bijectivity of the matching $ M$ implies that $\vert P\vert = \vert Q'\vert$. However, $n$ divides $\vert Q'\vert = P = \langle p_0\rangle_{\restr{f}{[0,1]}}$, so we obtain a contradiction with~\eqref{eq:p_one_definition}.
	
	Consider $p_1\in P'$. As $ M(p_1)\in Q'$, there is $q\in Q$ such that $\geometric{M(p_1)} = q$. Since $Q=\geometric{M(P)}$, $\geometric{M(p_2)}=q$, for some $p_2\in P$. Since $\geometric{p_1}\neq \geometric{p_2}$, $\Vert p_1 - p_2\Vert_\infty \geq \delta$, so that
	\begin{align*}
	\max_{p\in D(\restr{f}{[0,1]})\cup\Delta} \Vert \geometric{p}-\geometric{M(p)}\Vert_\infty
	&\geq \max(\Vert \geometric{p_1}-\geometric{M(p_1)}\Vert_\infty,\, \Vert \geometric{p_2}-\geometric{M(p_2)}\Vert_\infty) \\
	&\geq \frac{1}{2}(\Vert \geometric{p_1}-\geometric{q}\Vert_\infty + \Vert \geometric{p_2}-\geometric{q}\Vert_\infty)\\
	&\geq \frac{\delta}{2}.
	\end{align*}
	
	Since the bottleneck distance $\bottleneck(D(\restr{f}{[0,1]}), D(\restr{g}{[0,1]}))$ is an infimum over matchings $M$, it is also lower--bounded by $\delta/2$.
\end{proof}

For instance, taking the example from Figure~\ref{fig:delta_degenerate_example}, $\mu_{\restr{f_r}{[0,1]}}$ is supported on $\{(-1 -r, 1+r), (-1, 1)\}$, so that $\delta=\min(r, \frac{1}{2})$. However, the closest diagram of a degenerate function is $D=2\cdot\indicator_{(-1-r/2, 1+r/2)}$, which, by Proposition~\ref{prop:assumption_is_necessary} is realizable as the diagram of a function $g:[0,1]\rightarrow\R$. 
\begin{remark}
	In Section~\ref{sec:intro}, we mentioned the zero--crossings method in the context of estimating $N$. The number of times that the graph of a periodic function $f$ crosses $y=0$ can be read from the diagram. For a periodic function, a zero--crossing occurs between two local extrema, both of which have values different from zero. By periodicity, such a pair of extremal values generates exactly two zero--crossings. Hence, the number of zero--crossings is exactly $2\mu_f(\rbrack-\infty, 0\lbrack\times\rbrack0,\infty\lbrack)$.
\end{remark}

\section{Estimation of $N$ in the presence of noise}
\label{sec:method_multiplicity_points}
Suppose now that we observe $S=f\circ\gamma + W$, where $\Vert W\Vert <\epsilon$ is a continuous function $W:[0,1]\rightarrow \R$. Note that $W$ can have an infinite number of local extrema, so, in particular, the iteration over those extrema in Algorithm~\ref{alg:persistence} may not terminate.
In~\cite{chazal_structure_2016}, authors extend the definition of the persistence diagram and some of its properties to a more general setting, which does not require the function to have a finite number of critical points.
Since $[0,1]$ is a polyhedron and $S$ is continuous, the persistence diagram of $S$ can be defined~\cite[Theorem 2.22]{chazal_structure_2016}.
As an additional consequence, for any $\tau>0$, the number of points more persistent than in $\supp(\mu_S)\setminus \Delta_\tau$ is finite.

\begin{remark}
	Requiring the process $W$ to be continuous is reasonable in a stochastic setting. Under mild assumptions, thanks to the Kolmogorov theorem~\cite[Theorem 2.9]{le_gall_brownian_2016}, we can find a modification $\tilde{W}$ of $W$ which is almost--surely continuous.
\end{remark}

In the case when $W$ is a random process and each $W_t$ is a continuous random variable, the probability of extrema of $S$ having the same values is zero. Counting the multiplicity of points in $\mu_{\restr{f}{[0,1]}}$ introduced in Section~\ref{sec:model} needs to be adapted. We replace the notion of multiplicity by that of measures of neighborhoods. For a parameter $\tau>0$,
\begin{equation*}
\langle p\rangle_{S,\tau} = \mu_S(B(p,\tau))
\end{equation*}
is the number of points from $\mu_S$ in $B(p,\tau)$, the $\Vert\cdot\Vert_\infty$-ball of radius $\tau$ centered at $p$. We generalize the definition of $N(S)$ to
\begin{equation}
\label{eq:N_hat}
\hat{N}(S, \tau) = \gcd\{\langle p\rangle_{S,\tau} \mid p\in \supp(\mu_S)\setminus\Delta_\tau\}.
\end{equation}
Notice that in this second version, the $\gcd$ is computed only on points further than $\tau$ from the diagonal.
In the noise--free setting $W=0$, we have $\hat{N}(f\circ\gamma, 0) = N(f\circ\gamma)$.
In general, $\tau$ needs to be sufficiently large compared to the noise $\epsilon$, to account for the local displacement of points. It also needs to be bounded in terms of $\tau$, not to confound the global structure of the diagram.
We make those conditions explicit in Proposition~\ref{prop:N_stability}.
\begin{proposition}
	\label{prop:N_stability}
	Suppose that $\epsilon< \delta/6$ and that $f$ is non--degenerate. Then, for $\tau>0$ satisfying $2\epsilon<\tau<\delta/3$, we have that
	\begin{equation*}
	N(f\circ\gamma) = \hat{N}(S, \tau).
	\end{equation*}
\end{proposition}
We first state a lemma, which gives a stronger stability result.
\begin{lemma}
	\label{prop:balls_equality}
	Suppose that $\tau>0$ satisfies $2\epsilon<\tau<\delta/3$. Then, for any $p \in \supp(\mu_{f\circ\gamma})$ and $q\in B(p,\epsilon)$,
	\begin{equation}
	\label{eq:equality_of_balls}
	\mu_S(B(q,\tau)) = \mu_{f\circ\gamma}(B(p,\tau)).
	\end{equation}
\end{lemma}
\begin{proof}[Proof (Lemma~\ref{prop:balls_equality})]
	For any $0<r<\delta$, the separation of points in the diagram of $f$ ensures that
	\begin{equation*}
	\mu_{f\circ\gamma}(B(p,r))=\mu_{f\circ\gamma}(\{p\}).
	\end{equation*}
	Hence, $\mu_{f\circ\gamma}(B(p,\tau)) = \mu_{f\circ\gamma}(\{p\})$. By stability of the persistence diagram of $f\circ\gamma$, \eqref{eq:persistence_stability} writes $\bottleneck(\mu_S, \mu_{f\circ\gamma})\leq \Vert S -f\circ \gamma\Vert<\epsilon$. Let $M:(D(f\circ\gamma)\cup\Delta)\rightarrow (D(S)\cup\Delta)$ be a bijection of cost less or equal than $\epsilon+u,\ u>0$. Then, $\Vert M(p)-p\Vert_\infty\leq \epsilon + u$. Hence,  $\mu_{f\circ\gamma}(\{p\}) \leq \mu_S(B(p,\epsilon))$. As $d(p,q)<\epsilon$ implies $B(p,\epsilon)\subset B(q,2\epsilon)$, we get $\mu_S(B(p,\epsilon))\leq \mu_S(B(q,2\epsilon))\leq \mu_S(B(q,\tau))$.
	
	For the inequality in the other direction, by the same argument using stability as above, for any $u>0$, $\mu_S(B(q,\tau))\leq \mu_{f\circ\gamma}(M^{-1}(B(q,\tau)))\leq \mu_{f\circ\gamma}(B(q,\tau+\epsilon))\leq \mu_{f\circ\gamma}(B(p,\tau+2\epsilon))$.
	Since $\tau + 2\epsilon \leq 2\tau \leq \delta$, by \eqref{eq:equality_of_balls} with $r=\tau$ and $r=\tau+2\epsilon$ respectively, $\mu_{f\circ\gamma}(B(p,\tau + 2\epsilon))= \mu_{f\circ\gamma}(\{p\})= \mu_{f\circ\gamma}(B(p,\tau)).$
\end{proof}

\begin{proof}[Proof (Proposition~\ref{prop:N_stability})]
	Consider $q\in \supp(\mu_S)\setminus \Delta_\tau$. By stability of the persistence diagram, there is $p\in \supp(\mu_{f\circ\gamma})$, such that $q\in B(p,\epsilon)$. Proposition~\ref{prop:balls_equality} states that $\mu_S(B(q,\tau)) = \mu_{f\circ\gamma}(B(p,\tau)),$ so
	that $N(f\circ\gamma)$ divides $\hat{N}(S,\tau)$.
	For any $p\in \supp(\mu_{f\circ\gamma})$, by Proposition~\ref{prop:balls_equality},  $\mu_{f\circ\gamma}(B(p,\tau))=\mu_{S}(B(p,\tau))$, so that $\hat{N}(S,\tau)$ divides $N(f\circ\gamma)$.
	Thanks to the non--degeneracy of $f$ and by Proposition~\ref{prop:periodic_multiplicity}, we conclude that $\hat{N}(S,\tau)=N(f\circ\gamma)=N.$
\end{proof}

\begin{figure}
	\centering
	\includegraphics[width=0.65\textwidth]{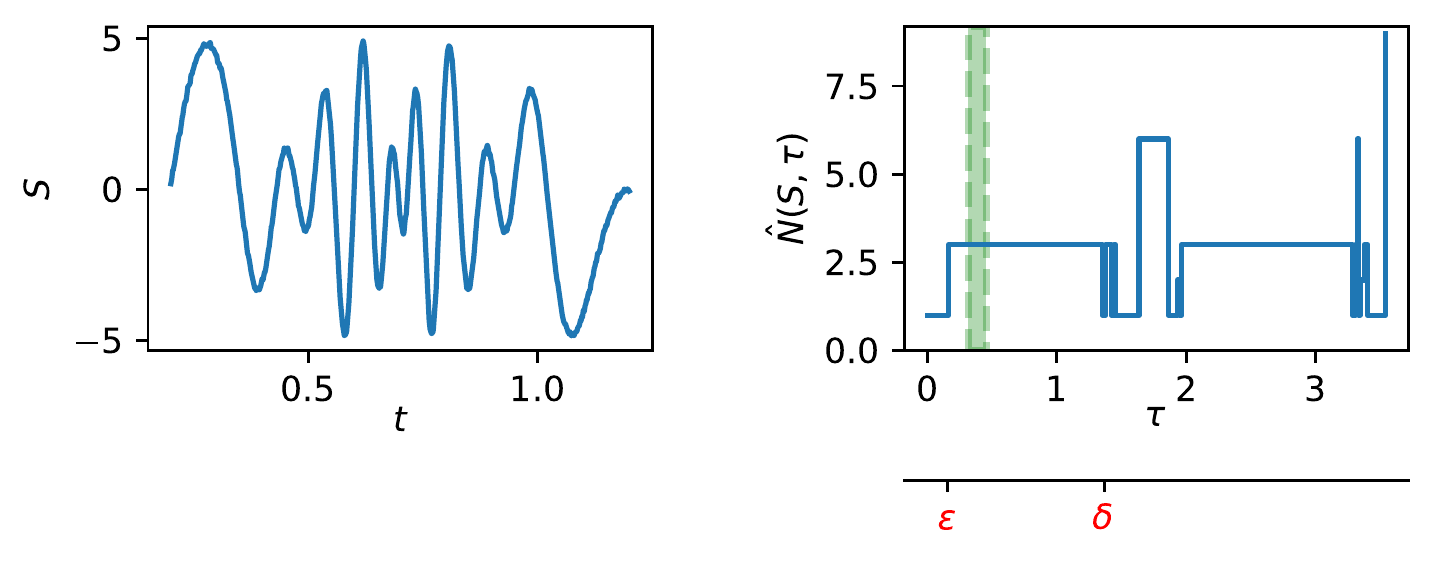}
	\caption{On the left, the graph of a sample from $S = f\circ\gamma + W$, with $\gamma(1)=3$, $\delta_f=1.37$ and $\Vert W\Vert_\infty\leq \epsilon=0.16$. On the right, the graph of the piecewise-linear function $\tau\mapsto \hat{N}(S,\tau)$. Marked in green is the guarantee from Proposition~\ref{prop:N_stability}.}
	\label{fig:N_stability_estimation}
\end{figure}
Figure~\ref{fig:N_stability_estimation} shows that the proposed estimator is correct, $\hat{N}(S,\tau)=3$, for $\tau\in [0.16, 1.35]$, which almost covers $[\epsilon,\delta]$. The guarantee given by Proposition~\ref{prop:N_stability} is much more pessimistic.

We propose to apply Proposition~\ref{prop:N_stability} in a random setting, with $(W_t)_{t\in[0,1]}$ being a collection of Gaussian random variables. One case of practical interest in signal processing is that of white noise. However, we cannot treat it in the continuous setting. We refer the reader to Appendix~\ref{app:sampled_signal}, where we first discretize the signal and then devise guarantees.
Instead, we propose to examine here the case of a regular Gaussian process. While the noise is now no longer bounded, we can calculate the probability that it remains bounded and that is the basis of our guarantee. We consider regular processes - that is, processes, which have a differentiable covariance function~\cite[section 4.3]{azais_level_2009}.

\begin{proposition}
	\label{thm:gaussian_correctness}
	Consider $S = f\circ\gamma + W(t)$, with $(W_t)_{t\in[0,1]}$ a Gaussian process with covariance function
	\begin{equation}
	\label{eq:noise_covariance}
	\Gamma(t)=\sigma^2\exp\left(-\frac{t^2}{2l^2}\right).
	\end{equation}
	Then, for any $\tau \leq \delta/3$,
	\begin{equation}
	\label{eq:corollary_azais_bound}
	P(\hat{N}(S, \tau) = N) \geq 1 - \left(\frac{1}{l^2\pi}\exp\left(\tfrac{-\kappa^2}{2}\right) + 2\phi\left(-\kappa\right)\right),
	\end{equation}
	where $\kappa = \frac{\tau}{2\sigma}$ and $\phi$ is the cumulative distribution function of a normal random variable.
\end{proposition}
\begin{proof}
	Proposition~\ref{prop:N_stability} states that $\hat{N}(S, \tau) = N$ whenever $S$ is $\tfrac{\tau}{2}$--close to $f\circ\gamma$, what translates to the following inclusion between the events
	\begin{equation*}
	\{\Vert f\circ\gamma-S\Vert_\infty \leq \tfrac{\tau}{2}\} \subseteq \{\hat{N}(S, \tau) = N\}.
	\end{equation*}
	Define $Z=\frac{W}{\sigma}$, where $\sigma = \sqrt{\Gamma(0)}$. Then,
	\begin{align*}
	P(\Vert f\circ\gamma-S\Vert_\infty \leq \tfrac{\tau}{2})
	&= P\left(\sup_{t\in [0,1]}\vert W_t\vert \leq \tfrac{\tau}{2}\right) \\
	&= P\left(\sup_{t\in [0,1]}\vert Z_t \vert \leq \tfrac{\tau}{2\sigma}\right).
	\end{align*}
	Set $u=\frac{\tau}{2\sigma}$. Since $\sup_{t\in [0,1]} \vert Z_t \vert = \max(\sup_{t\in [0,1]} Z_t,\  \sup_{t\in[0,1]}-Z_t)$, the events $\{\sup_{t\in [0,1]} Z_t>u\}$ and $\{\sup_{t\in [0,1]} -Z_t > u\}$ are a cover of $\sup_{t\in [0,1]} \vert Z_t \vert >u$.
	\begin{align*}
	P\left(\sup_{t\in [0,1]}\vert Z_t \vert \leq \tfrac{\tau}{2\sigma}\right) &= 1- P(\vert \sup_{t\in [0,1]} Z_t \vert > u), \\
	&\geq 1- (P(\sup_{t\in [0,1]} Z_t > u) + P(\sup_{t\in[0,1]}(-Z_t)> u))\\
	&\geq 1- 2P(\sup_{t\in [0,1]} Z_t > u),
	\end{align*}
	where the ultimate inequality follows from the fact that $Z(t)$ and $-Z(t)$ have equal distributions.
	Let $r(s_1,s_2) = \mathbb{E}[Z_{s_1} Z_{s_2}]$ and $r_{1,1}(s, t) = \tfrac{\partial^2}{\partial s_1 \partial s_2}r(s,t)$.
	By \cite[Proposition 4.1]{azais_level_2009},
	\begin{equation}
	\label{eq:azais_bound}
	P\left(\sup_{t\in [0,1]} Z_t > u\right) \leq \frac{\exp(-u^2/2)}{2\pi}\int_0^1 \sqrt{r_{1,1}(t,t)}\diff t + 1 - \phi(u).
	\end{equation}
	Here, $r(s_1, s_2) = \Gamma(s_1-s_2) = \exp\left(-\frac{(s_1-s_2)^2}{2l^2}\right)$, so computing the derivative of the covariance function
	\begin{equation*}
	r_{1,1}(s, t) = \frac{\partial^2}{\partial s_1 \partial s_2}r(s,t) = \frac{s-t+1}{l^2}\exp\left(-\frac{(s-t)^2}{2l^2}\right).
	\end{equation*}
	For $s=t$,
	\begin{equation*}
	r_{1,1}(t,t) = \frac{1}{l^2}.
	\end{equation*}
	Computing the right--hand side of \eqref{eq:azais_bound} and putting it all together, we obtain
	\begin{equation*}
	P(\hat{N}(S, \tau) = N) \geq 1 - \left(\frac{1}{l^2\pi}\exp\left(\tfrac{-\kappa^2}{2}\right) + 2\phi\left(-\kappa\right)\right).
	\end{equation*}
\end{proof}
Let us analyze the bound \eqref{eq:corollary_azais_bound}. The parameter $l$ in~\eqref{eq:noise_covariance} quantifies the horizon of dependence of the stochastic process. The bound is increasing in $l$, implying that a long dependence (large $l$) yields a simpler structure for the method. In the limit $l\rightarrow \infty$, the expression on the right, $P(\hat{N}(S, \tau) = N) \geq 1 - 2\phi\left(-\kappa\right)$, is the the same as the probability that $\Vert (f\circ\gamma)(0)-S(0)\Vert_\infty<\tau/2$. On the other hand, when the interactions become short-term ($l$ small), the bound becomes trivial: there is a constant $l_0$ depending on $\kappa$, such that for $l\in\rbrack 0, l_0\rbrack$, the bound is 0. The other governing parameter is $\kappa$ - the ratio between the chosen scale $\tau$ and the standard deviation $\sigma$ of the process.
So, since $\tau\leq \frac{\delta}{3}$ has to be smaller than a fraction of the separation of $f$, $\kappa\leq \frac{\delta}{6\sigma}$ is also bounded. Regardless how good of a method to choose $\tau$ we have, we obtain the best lower--bound of the probability of correctness only in terms of the quantities characteristic of the signal and the stochastic process. Experimental observations regarding $l$ are discussed in Section~\ref{sec:numerical_results}.

\section{Inference of an odometric sequence}
\label{sec:method_odometry}
In this section, we propose an odometric sequence for Problem~\ref{prob:odometry}, which we take as a subset of the local minima of the signal, chosen using the estimator of $N$ proposed in Section~\ref{sec:method_multiplicity_points}. We first examine the noiseless case $f\circ\gamma$, with $\gamma:[0,1]\rightarrow [0,N]$. Let $\mathcal{C}$ be the set of local minima of $f\circ\gamma$. If $K$ is the number of local minima in a single period of $f$, then $\vert \mathcal{C}\vert = NK$. By picking every $K$-th minimum in $\mathcal{C}$, we obtain an odometric sequence. More precisely, we lexicographically index the ordered minima $\mathcal{C}=(t_{n,k})_{\substack{n=1,\ldots,N\\k=1,\ldots,K}}$ and then, for any $k=1,\ldots,K$, and $n=2,\ldots, N$, we have $\gamma(t_{n,k}) = \gamma(t_{n-1,k}) + 1$.
\begin{remark}
	Notice that since $N$ is part of the data in Problem~\ref{prob:odometry}, we do not even need to assume that $f$ is non--degenerate. If we wanted to use the same technique with $N$ unknown, we would need to estimate it and the non--degeneracy of $f$ is necessary for Proposition~\ref{prop:N}.
\end{remark}

\begin{figure}
	\centering
	\includegraphics[width=0.33\textwidth]{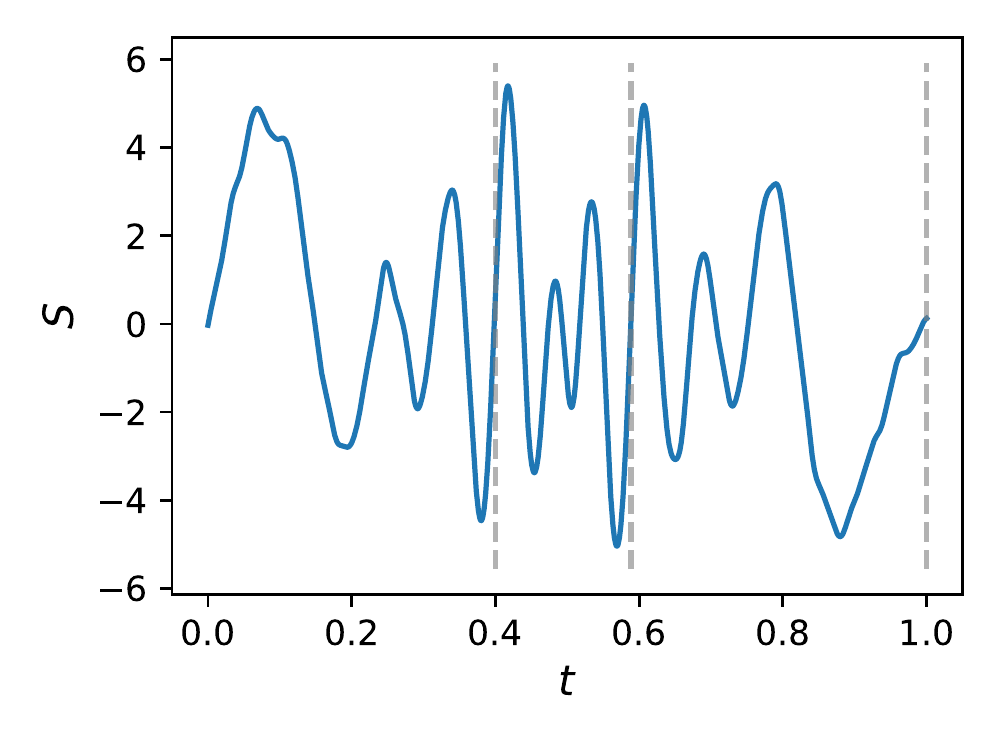}\hfill
	\includegraphics[width=0.33\textwidth]{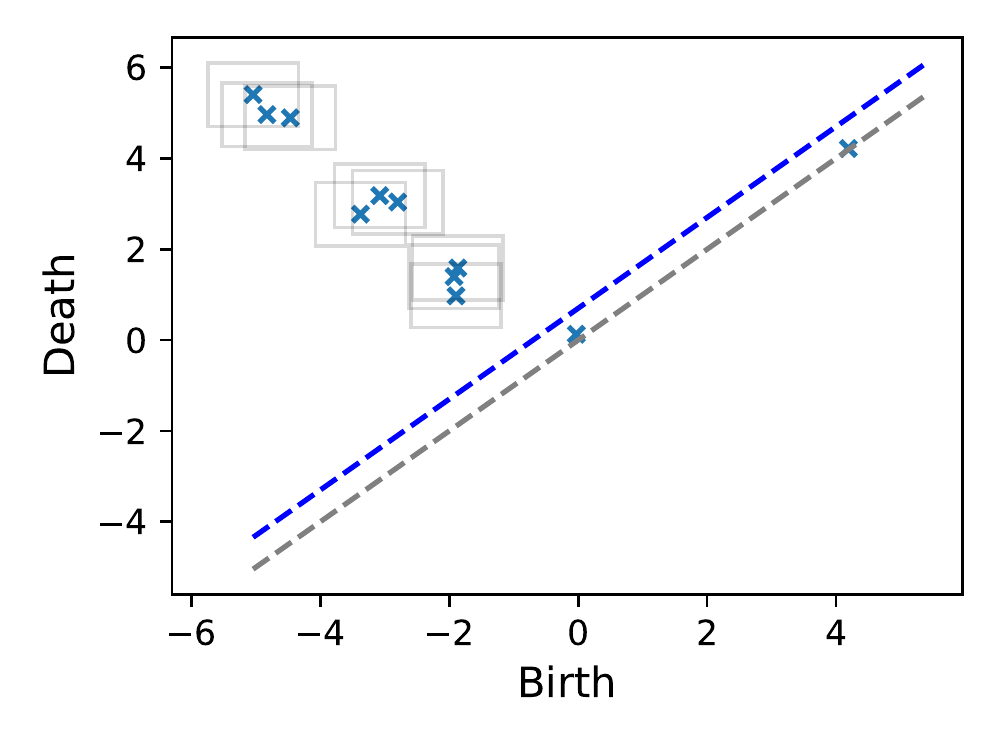}\hfill
	\includegraphics[width=0.33\textwidth]{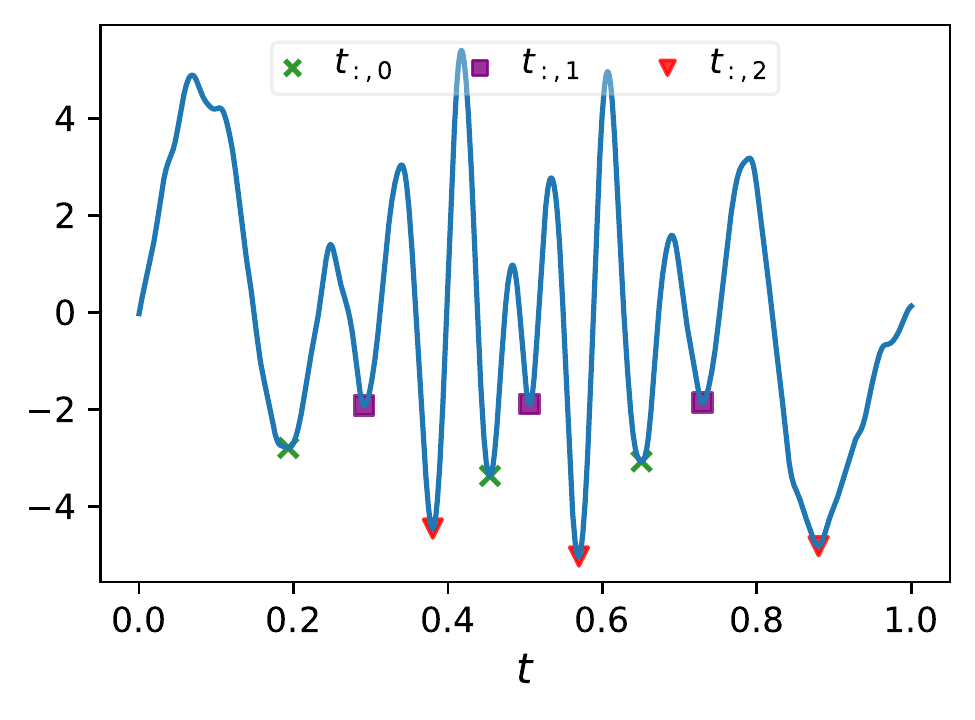}
	\caption{Example of an odometric sequence. On the left, an observed signal, with $N=3$. In the center figure, its persistence diagram, with $\Vert\cdot\Vert_\infty$-balls of radius 0.7, marking the sets whose measure is evaluated. On the right, the observed signal with the odometric sequences stemming from persistent local minima. For each $k=0, 1,2$, $(\hat{t}_{n,k})_{n=1}^3$ is an odometric sequence.}
	\label{fig:odometry}
\end{figure}

Let us consider the case of a signal corrupted by noise $S=f\circ\gamma + W$. While the location of the local minima of $S$ is slightly different than those of $f\circ\gamma$, we can expect to be able to find a set of minima, in one-to-one correspondence with minima of $\mathcal{C}$. Thus, the odometric property will not be exactly satisfied. An example is shown in Figure~\ref{fig:odometry}.
We observe a signal, where $\gamma(1)=3$. Note that two points lie between the diagonal and it's offset, marked by the dashed lines. These are not used in the estimation of $N$ and the corresponding local minima are ignored when selecting the odometric sequences, shown on the right.
Proposition~\ref{thm:odometry} makes precise the guarantees on the existence of the sequence. It also states that it satisfies the odometric property, up to a constant which depends on $f,\tau$ and $\epsilon$.

\begin{proposition}
	\label{thm:odometry}
	Let $\tau>0$ and $\mathcal{\hat{C}}_\tau$ be the set of local minima of $S$, corresponding to points in the diagram with persistence more than $\tau$.
	If $\tau\in \rbrack 2\epsilon, \delta/3\lbrack$, then
	\begin{equation*}
	\vert \mathcal{\hat{C}}_\tau\vert = NK.
	\end{equation*}
	In addition, if we order the minima $\mathcal{\hat{C}}_\tau = \{t_{1,1},\ldots t_{1,K}, t_{2,1},\ldots, t_{N,K}\}$, then
	\begin{equation*}
	\vert\gamma(\hat{t}_{n,k}) - \gamma(\hat{t}_{n-1,k}) - 1 \vert \leq 2R(\tau + 2\epsilon), 
	\end{equation*}
	where 
	\begin{equation*}
	R(\nu) = \sup_{x\in \gamma(\mathcal{C})}\ \inf_{r>0}\{r\mid f(x+r)-f(x)>\nu,\ f(x-r)-f(x)>\nu\}
	\end{equation*}
	is a constant depending on $f$ and independent of $\gamma$ and $N$.
\end{proposition}

Before we proceed to the proof, we recall a result from the theory of persistence: between any two consecutive local minima of $f$, there is a local maximum, larger than those minima by at least $\delta$. We provide a proof in Appendix~\ref{app:min_max_min_lemma}.
\begin{lemma}
	\label{lemma:min_max_min}
	Let $x_1< x_2$ be two, consecutive local minima of $f$ and suppose that the separation constant \eqref{cond:separation} for $\mu_f$ is $\delta>0$. Then,
	\begin{equation*}
	[\min(f(x_1), f(x_2)), \max(f(x_1), f(x_2)) + 2\delta] \subset f([x_1, x_2]).
	\end{equation*}
\end{lemma}

\begin{proof}[Proof of Proposition~\ref{thm:odometry}]
	When $2\epsilon<\tau<\delta/3$,
	\begin{align*}
	\vert\mathcal{\hat{C}}_\tau\vert &= \mu_S(\R^2\setminus\Delta_\tau)\\
	&= \mu_{f\circ\gamma}(\R^2\setminus \Delta)\\
	&= \vert \mathcal{C}\vert\\
	&= NK,
	\end{align*}
	where the second inequality follows from Lemma~\ref{prop:balls_equality}. This yields the first statement.
	\\	
	For the second part, let us start by defining $x_k = \gamma(t_{1,k}) \in \lbrack 0, 1\rbrack.$
	By definition, $\gamma(t_{n,k}) - \gamma(t_{n-1,k}) = x_k + (n-1) - (x_k + (n-2)) = 1$, so	
	\begin{align*}
	\vert\gamma(\hat{t}_{n,k}) - \gamma(\hat{t}_{n-1,k}) - 1\vert
	&= \vert\gamma(\hat{t}_{n,k}) - \gamma(\hat{t}_{n-1,k}) - (\gamma(t_{n,k}) - \gamma(t_{n-1,k}))\vert\\
	&\leq \vert\gamma(\hat{t}_{n,k}) - \gamma(t_{n,k})\vert + \vert\gamma(\hat{t}_{n-1,k}) - \gamma(t_{n-1,k}))\vert.
	\end{align*}
	Let us introduce $R_k = \max (R_k^-, R_k^+)$, with
	\begin{equation}
	\label{eq:r_k_minus_plus}
	\begin{aligned}
	R_k^- &= \sup_r\{r \mid \forall y\in [x_k-r, x_k],\ f(y)\leq f(x_k) + \tau + 2\epsilon \},\\
	R_k^+ &= \sup_r\{r \mid \forall y\in [x_k, x_k+r],\ f(y)\leq f(x_k) + \tau + 2\epsilon\},
	\end{aligned}
	\end{equation}
	which upper--bounds how far one needs to look in either direction from $x_k$, so that the graph of $f$ leaves the interval $[f(x_k), f(x_k)+\delta]$. We now need to show that, for all $n,k$,
	\begin{equation*}
	\vert \gamma(\hat{t}_{n,k}) - \gamma(t_{n,k})\vert \leq \max\{R_k^+, R_k^-\}.
	\end{equation*}
	\begin{figure}
		\centering
		\includegraphics[width=0.99\textwidth]{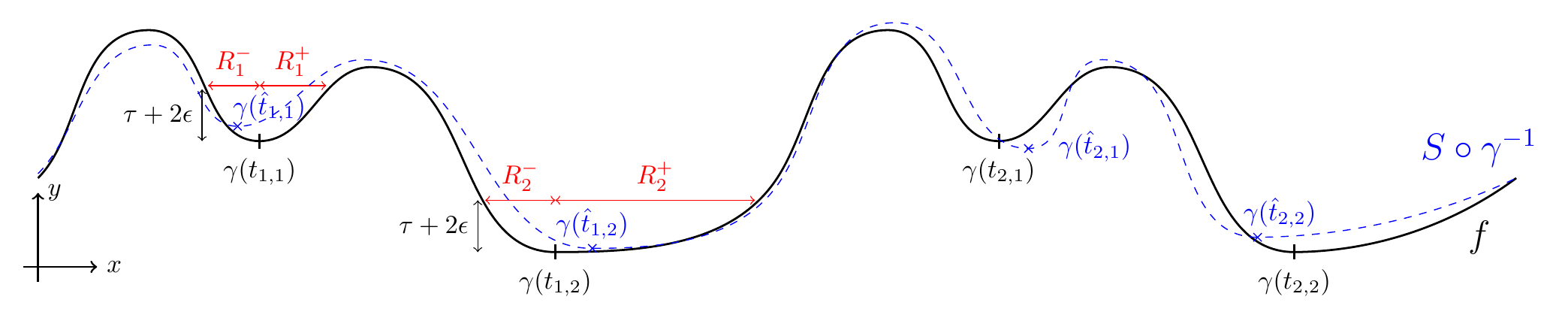}
		\caption{The figure shows an example of the graphs of $f$ and $S\circ\gamma^{-1}$, with local minima of both functions and $R_k^-, R_k^+$ from~\eqref{eq:r_k_minus_plus} illustrated.}
		\label{fig:landmark_stability}
	\end{figure}
	More specifically, we show that for any $n,k$, $\vert \mathcal{\hat{C}}_\tau \cap I_{n,k}\vert= 1$. Since there are exactly $\vert \mathcal{\hat{C}}_\tau\vert = NK$ intervals $I_{n,k}$, $\hat{t}_{n,k} \in I_{n,k}$. We illustrate this in Figure~\ref{fig:landmark_stability} and conclude by setting $R(\tau + 2\epsilon) = \max_{k}\{R_{k}^-, R_{k}^+\}.$
	\\
	We now show $\vert \mathcal{\hat{C}}_\tau \cap I_{n,k}\vert= 1$. Observe that since $S$ is continuous, it has a minimum $a^*$ over $I_{n,k}$. This minimum is not achieved at the extremities of $I_{n,k}$, since
	\begin{align*}
	S(\gamma^{-1}(\gamma(t_{n,k}) + R_k^+)) - S(t_{n,k}) \geq \tau + 2\epsilon + (W_{\gamma^{-1}(\gamma(t_{n,k}) + R_k^+)} + W_{t_{n,k}}) > \tau. \\
	S(\gamma^{-1}(\gamma(t_{n,k}) - R_k^-)) - S(t_{n,k}) \geq \tau + 2\epsilon + (W_{\gamma^{-1}(\gamma(t_{n,k}) - R_k^-)} + W_{t_{n,k}}) > \tau.
	\end{align*}
	The inequalities above also show that the $a^*$ has a persistence greater than $\tau$, so $a^*\in \mathcal{\hat{C}}_\tau$.
	\\
	Suppose now that two elements $a^* < b^*\in\mathcal{\hat{C}}_\tau$ are contained in the same interval $I_{n,k}$. Lemma~\ref{lemma:min_max_min} states that there is a local maximum $c^*\in \rbrack a^*, b^*\lbrack \subset I_{n,k}$. If $t_{n,k}< c^*$, $S(c^*) - S(a^*)>\delta$ contradicts the minimality of $R_k^+$ and therefore $a^*=b^*$. If $t_{n,k}>c$, we obtain a contradiction of the minimality of $R_k^-$ since $S(c^*) - S(b^*)>\delta$.
\end{proof}

As we would expect, the more strictly convex $f$ is around prominent local minima, and the smaller $\tau$ and $\epsilon$, the smaller $R_k$ is. It is interesting to note that the guarantee is not cumulative: more specifically, for any $0\leq n,m\leq N$, we have
\begin{equation*}
\vert \gamma(\hat{t}_{n,k}) - \gamma(\hat{t}_{m,k}) - (n-m)\vert \leq 2R(\tau+2\epsilon).
\end{equation*}

\section{Numerical experiments and applications}
We introduce more robust versions of $\hat{N}$ and we quantify their performance on synthetic signals. Then, we test the odometric sequence method on real, magnetic data. We measure the quality of that sequence using positioning data.

\subsection{Practical adaptations of $\hat{N}$}
\label{sec:pratical_adaptations}
We introduce $\hat{N}_c$, obtained by replacing in~\eqref{eq:N_hat} the number of points in neighborhoods by the sizes of some partitions.
Specifically, let $\Partition_\tau$ be a partition of the persistence diagram $D(S) \cup \Delta$ obtained via single--linkage hierarchical clustering with a scale-parameter $\tau$. We define
\begin{equation}
\label{eq:N_hat_clustering}
\hat{N}_c(S,\tau) = \gcd\{\mu_S(\partition)\mid \partition\in\Partition_\tau, \partition\cap\Delta=\emptyset\}.
\end{equation}
Let us now consider the map $h_S: \tau\mapsto \hat{N}_c(S,\tau)$, for $\tau\in \R_+^*$, with the convention that $\gcd(\emptyset)=1$. The map $h_S$ is piece--wise constant and, since the persistence of all points is bounded by $(\max(S)-\min(S))/2$, $\restr{h_S}{\rbrack(\max(S)-\min(S)), \infty\lbrack} =1$. Let $\mathcal{I}_n$ be the collection of maximal intervals, where $h_S(\tau)=n$. Finally, we set $\hat{N}_c^T$ to the value which is realized over the longest interval 
\begin{equation}
\hat{N}_c^T(S) = \argmax_{n>1} (\max_{I\in \mathcal{I}_n} \length(I)).
\end{equation}
To compute the single--linkage clustering, we use the scikit--learn library~\citep{scikit-learn}.

Let us motivate the introduction of $\hat{N}_c$.
Experiments show that points in the diagrams $\mu_S$ form non--overlapping clusters. The diameters of those clusters can be greater than the distance which separates them (when the condition $6\epsilon<\delta$ is not satisfied). In such a situation, the estimator~\eqref{eq:N_hat} will be incorrect, even if, to the human eye, the clusters are easily identifiable. 
In the proof of Lemma~\ref{prop:balls_equality}, we notice that if the upper bound of noise $\epsilon$ is smaller than the separation $\delta/3$, the points in the persistence diagram are partitioned into neighborhoods of size $\tau$. In particular, if two points in the diagram share a neighbor at distance no greater than $\tau$, their $\tau$-neighborhoods are equal due to the separation property satisfied by $f$. This is exactly the property satisfied by a partition produced by single--linkage clustering. Namely, that algorithm produces the coarsest partition $\Partition_\tau$ of $\R^2$ such that: for any $p_1,\,p_2\in D(S)$, if there exists $q\in D(S)$ such that $d(p_1,q)<\tau$ and $d(q, p_2)<\tau$, then there exists $\partition\in\Partition$ such that $p_1,p_2\in \partition$.
This motivates replacing $(B(p, \tau))_{p\in \supp(\mu_S\setminus\Delta_\tau)}$ with a partition $\Partition_\tau$ from a single-linkage algorithm, leading to an algorithm more robust to noise and to the choice of $\tau$.

The second estimator $\hat{N}_c^T$ allows us to avoid the selection of $\tau$: it takes the value for which $h_S$ is constant for the longest interval.
Suppose there exists $\tau_0>0$ such that $h_S(\tau_0)=N$ and that no point in $\partition_\tau$ is isolated. Then $N$ divides $h_S(\tau)$, for all $\tau\in\rbrack\tau_0,\tau_{\max}\lbrack$, where $\tau_{\max}$ is the smallest value of $\tau$ such that $\partition_\tau$ contains a single cluster, intersecting the diagonal.
By a similar argument as in Proposition~\ref{prop:N_stability}, $h_S(\tau) = N$ for $\tau \in \rbrack 2\epsilon, \delta/3\lbrack$. So, there exists $I\in\mathcal{I}_N$ such that $\rbrack 2\epsilon, \delta/3\lbrack\subset I$. Usually, this inclusion will be strict, because the clustering scheme is less sensitive to $\epsilon$.

The automatic choice of $\tau$ comes with some limitations. Since we set the domain of $h_S$ in a data--driven way to $\lbrack 0,(\max S - \min S)/2\rbrack$, we encounter problems when $\delta$ is small compared to the amplitude of the signal. It is particularly visible when the signal $f$ is close to a degenerate--function. Consider $f_r$ from Figure~\ref{fig:delta_degenerate_example}. Even in the noiseless setting, $\delta \leq \Vert(-1-r, 1+r) - (-1, 1)\Vert_\infty=r$, so $h_{\restr{f_r}{[0,1]}}(\tau) = 1$ for $\tau<r$ and $h_{\restr{f_r}{[0,1]}}(\tau) = 2$ for $r\leq \tau\leq 2(1+r)$. Therefore, $\hat{N}_c^T(\restr{f_r}{[0,1]}) = 1$ if and only if $r>1$. On the other hand, such behavior is to be expected from any method of automating the choice of $\tau$: estimating $N$ amounts to looking for symmetries in the diagram of $S$, so if the true symmetry is present only for a narrow range of scales $\tau$, it will be harder to capture that scale correctly and automatically.

Ultimately, we are often interested in computing the odometric sequence for which a value of $\tau$ is needed. We take $\tau$ to be the middle of the longest interval on which $h_S$ is constant, so that $\hat{N}_c(S,\tau) = \hat{N}_c^T(S)$.

\subsection{Estimating $N$ on synthetic data}
\label{sec:numerical_results}
\begin{figure}
	\centering
	\includegraphics[width=0.99\textwidth]{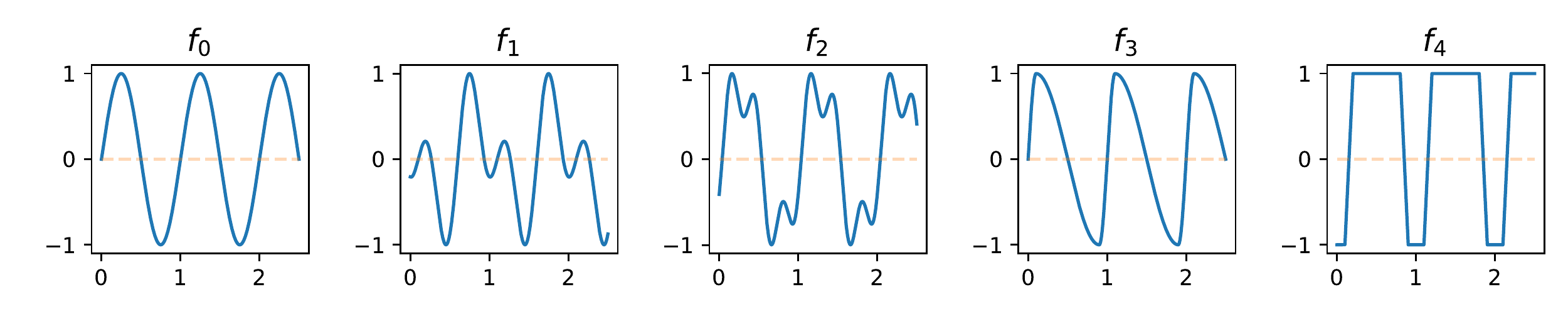}
	\caption{The graphs of five one--periodic functions over $[0, 2.5]$.}
	\label{fig:template_functions}
\end{figure}
We choose five template functions $f_0,\ldots, f_4$ and we picture their periods in Figure~\ref{fig:template_functions}. The functions $f_0,\, f_3$ and $f_4$ all have a single pair of extrema per period, while $f_1$ and $f_2$ have two and three pairs per period respectively. All the functions are normalized so that their range is $[-1,1]$, but their separation parameters $\delta$ differ: $\delta_0= 1=\delta_3=\delta_4$, $\delta_1=0.21$ and $\delta_2=0.13$. We generate 100 reparametrizations $\gamma$ by sampling $N_k$ uniformly from $\{5,\ldots,50\}$ and sampling $(t_l^k)_{l=1}^{N_k}$ uniformly in $[0,1]$. Ordered ascendingly, they define the starts of the periods. This yields a collection of $500$ signals, that we perturb with noise $(W_t)_{t\in[0,1]}$, a Gaussian process with covariance as in~\eqref{eq:noise_covariance}, with parameters $\sigma\in [0.0001, 6]$ and $\gamma\in [0.01, 0.4]$.

\begin{figure}
	\centering
	\includegraphics[width=0.99\textwidth]{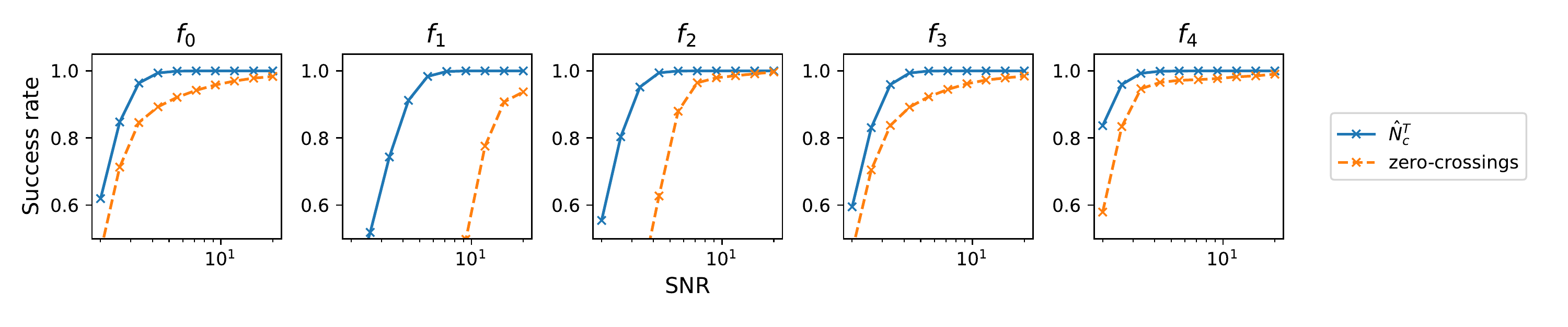}
	\includegraphics[width=0.99\textwidth]{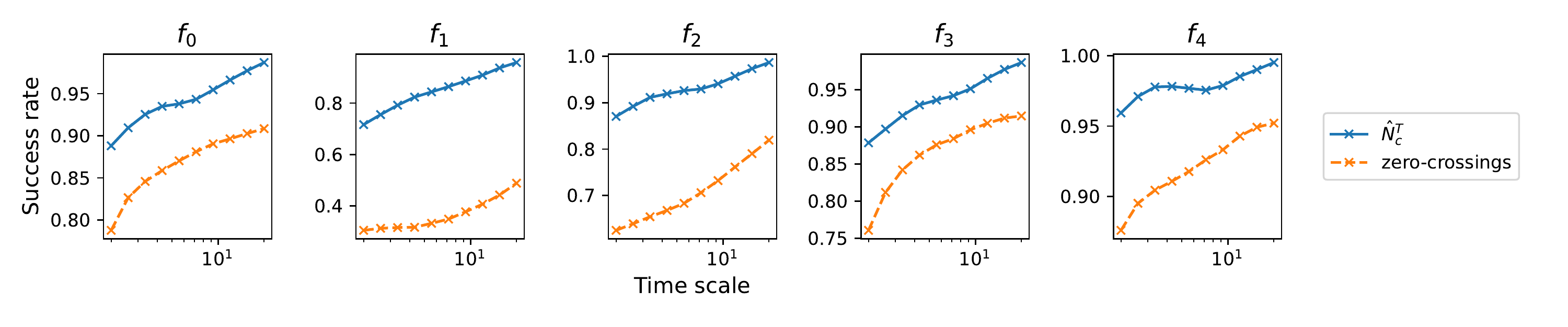}
	\caption{The success rates for estimating $N$, while varying the noise level $\sigma$ (top row) and the time--scale $l$ (bottom row) of the Gaussian process $W$ with covariance~\eqref{eq:noise_covariance}.}
	\label{fig:synthetic_results}
\end{figure}
We compare two estimators. We define an oracle, based on the zero crossings method, by specifying the number of zero crossings in a period of $f_k$ - information that is inaccessible in practice. We benchmark $\hat{N}_c^T$ against this zero crossings oracle, with the fraction of samples for which $N$ was estimated correctly as a metric. We present the results for each template function $f_k$ in Figure~\ref{fig:synthetic_results}. Globally, the performance of both estimators is increasing as the SNR increases ($\sigma$ decreases). The estimator $\hat{N}_c^T$ outperforms the oracle in all scenarios. The sensibility of the zero--crossings method to local extrema close to zero is confirmed by the poor performance on $f_1$: a noise of small amplitude can create additional zero--crossings. Even though $\hat{N}_c^T$ also shows a degraded performance for this periodic function, it outperforms the benchmark and that by the biggest margin. The difference in performance between the two methods is also most noticeable on $f_2$. At high noise levels, the pairs of local and global extrema get identified as the same, as it would be the case for a degenerate function.

\begin{figure}
	\centering
	\includegraphics[width=0.99\textwidth]{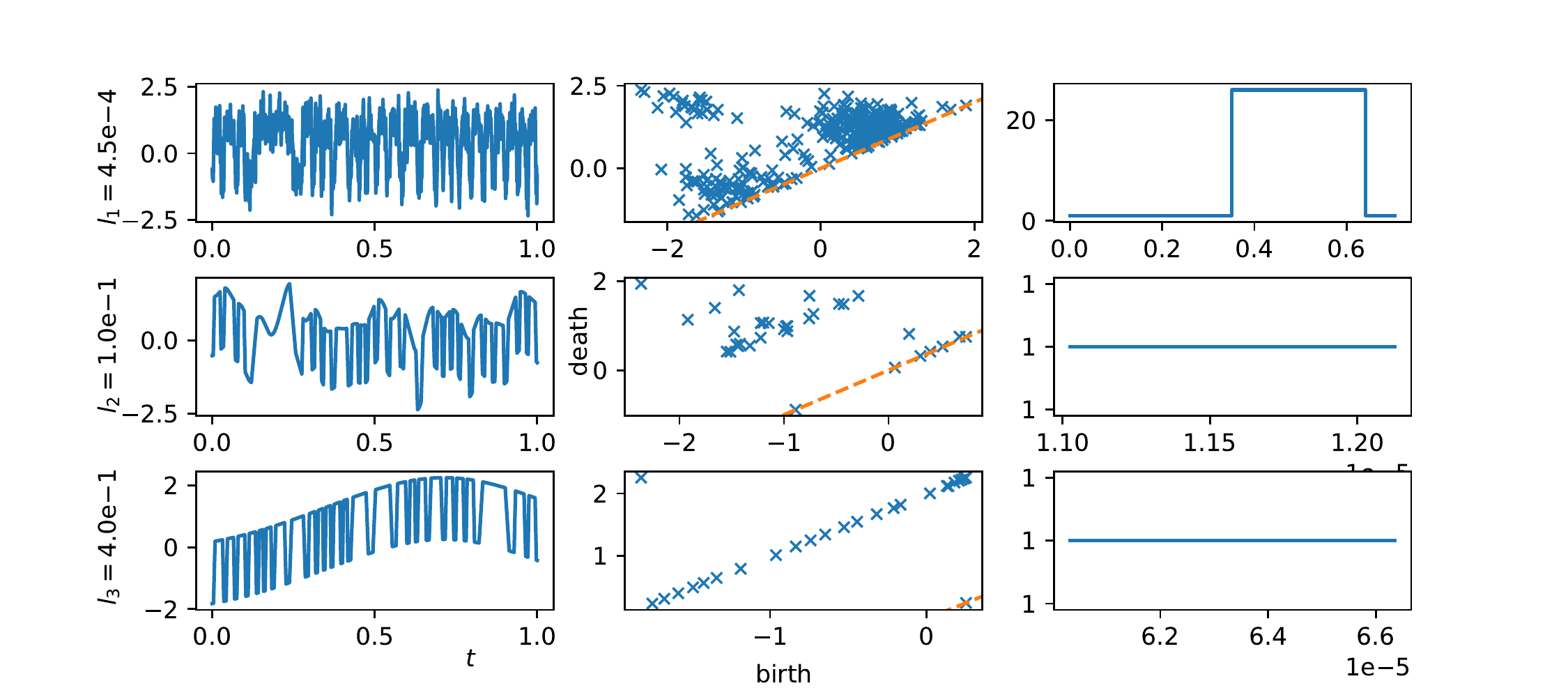}
	\caption{The left column represents realizations of $S=(f_4\circ\gamma)(t) +W_t$, for $W_t$ sampled with different time--scales $l_1,l_2,l_3$. From left to right, each row contains the graph of $S$, the diagram $\mu_S$ and $h_S$. 
	}
	\label{fig:explain_timescale}
\end{figure}
Since $\max\vert W_t\vert$ is decreasing in the time--scale $l$, we expect the success rate to be increasing in $l$.
The success rate of $\hat{N}_c^T$ is increasing, but starts to decrease and reaches a minimum at around $l=0.1$. We believe that the initial increase is indeed due to the decrease in $\Vert W_t\Vert_\infty$. In this range of $l$, the noise is high--frequency and of small amplitude, leading to many points in the diagram, as shown in the top row in Figure~\ref{fig:explain_timescale}. The difficulty in estimating $N$ comes essentially from the presence of many points close to the diagonal. However, the points corresponding to minima of $f$ are generally separable from those spurious points. When the time scale increases to around 0.05 - 0.1, the noise has lower frequencies, almost in the range of $\gamma'$. Looking at the graph on the left, the noise shifts some periods vertically, creating isolated points in the persistence diagram. In the output of a clustering algorithm those points are isolated, instead of being clustered with the rest. When $l$ increases even further, the induced oscillations are generally of smaller amplitude and resemble a slow drift of the period. In the diagram, such a drift manifests itself as a densely--sampled line, which can be identified as a single cluster.

\subsection{Applications to magnetic odometry}
\label{sec:magnetic_odometry}
Problems~\ref{prob:estimate_N} and \ref{prob:odometry} are motivated by an industrial application: magnetic odometry. The magnetic field $S$ measured inside a car and close to one of its wheels exhibits recurrent structure, periodic in the rotation of that wheel. The function $f(x)$ is the value of the magnetic field when the wheel is in a position $x\in\R$ so $\gamma$ encodes the rotations of the wheel over time. Inferring $\gamma$ from $S$ alone would allow one to compliment the inertial and GPS--based position estimates with an independent estimation, resulting in a more robust estimator~\citep{bristeau_techniques_2012}. From an odometric sequence, we define a position estimate $\hat{\gamma}(t) = \sum_{k} 1_{t_k\leq t}$.

\begin{figure}
	\centering
	\includegraphics[width=0.75\textwidth]{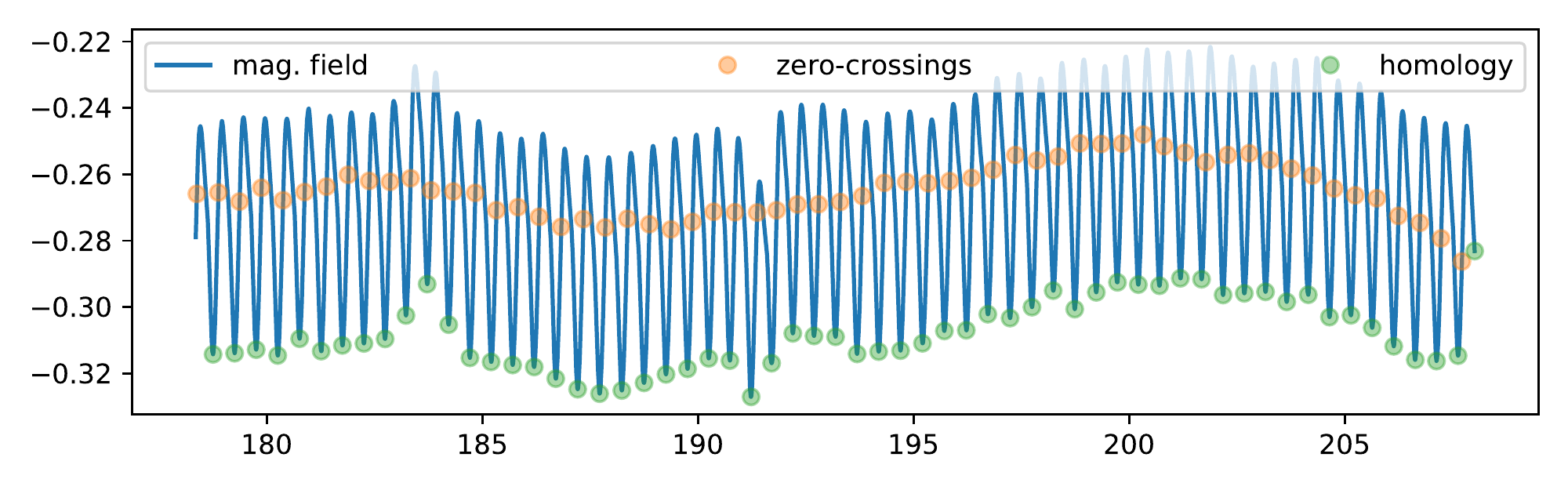}
	\caption{Magnetic field recorded in a vehicle moving in a straight line, at a constant speed.}
	\label{fig:magnetic_odometry}
\end{figure}

The magnetic signal is sampled at 125Hz. As preprocessing, we detrend the signal by removing the median, calculated using a 5--second sliding window (ie, $5\times125$ elements). By applying the zero crossings algorithm, we obtain the reference sequence. The topological method consists of applying the proposed estimator $\hat{N}_c^T$ and retrieving the associated sequence. Figure~\ref{fig:magnetic_odometry} shows the magnetic signal recorded in a vehicle moving in a straight line and the resulting sequences. Both capture the periodicity well and satisfy the odometric property~\eqref{eq:odometry_condition}.

We propose a metric to assess the quality of a solution to the odometric problem. Given a sequence $(t_n)_{n=1}^N$ and a reference $\tilde{\gamma}(t)$, we examine the distribution of $d_n = \tilde{\gamma}(t_n) - \tilde{\gamma}(t_{n-1})$ for $n=2,\ldots,N$. First, we note that for a true odometric sequence, all $d_n$ are equal to the circumference of the wheel, $C=1.94$ meters here. Therefore, we use two criteria. First, we quantify the number of outliers in three different ways: the number of outliers measured by the number of too short segments (TS), the number of too large segments (TL) and the ratio of ``correct" segments (CR)
\begin{equation*}
TS=\vert\{n\mid d_n<0.9C\}\vert,\qquad TL=\vert\{n\mid d_n>1.1C\}\vert,\qquad CR = 1 - \frac{TS+ TL}{N}.
\end{equation*}
Second, we use $R((d_n)_{n=1}^N)$ the mean deviation of $d_n$ from the circumference as the measure of dispersion. We note that the two criteria are complimentary. Indeed, TS and TL quantify whether the number of turns of the wheel is estimated correctly, on subsegments of the signal. The dispersion measures whether the odometric sequence captures turns precisely, or if the locations are only approximately correct. Note that a perfect reference $\gamma$ is not available, so the errors in the displacement estimation $\tilde{\gamma}$ used as a reference will induce some additional variance. 

We use a 16-minute recording of a vehicle moving in traffic on a public road network. The magnetometer is placed close to the right rear wheel (in a favorable position) and the reference displacement estimation is provided by a combination of the GPS and inertial sensors.
\begin{figure}
	\centering
	\includegraphics[width=0.5\textwidth]{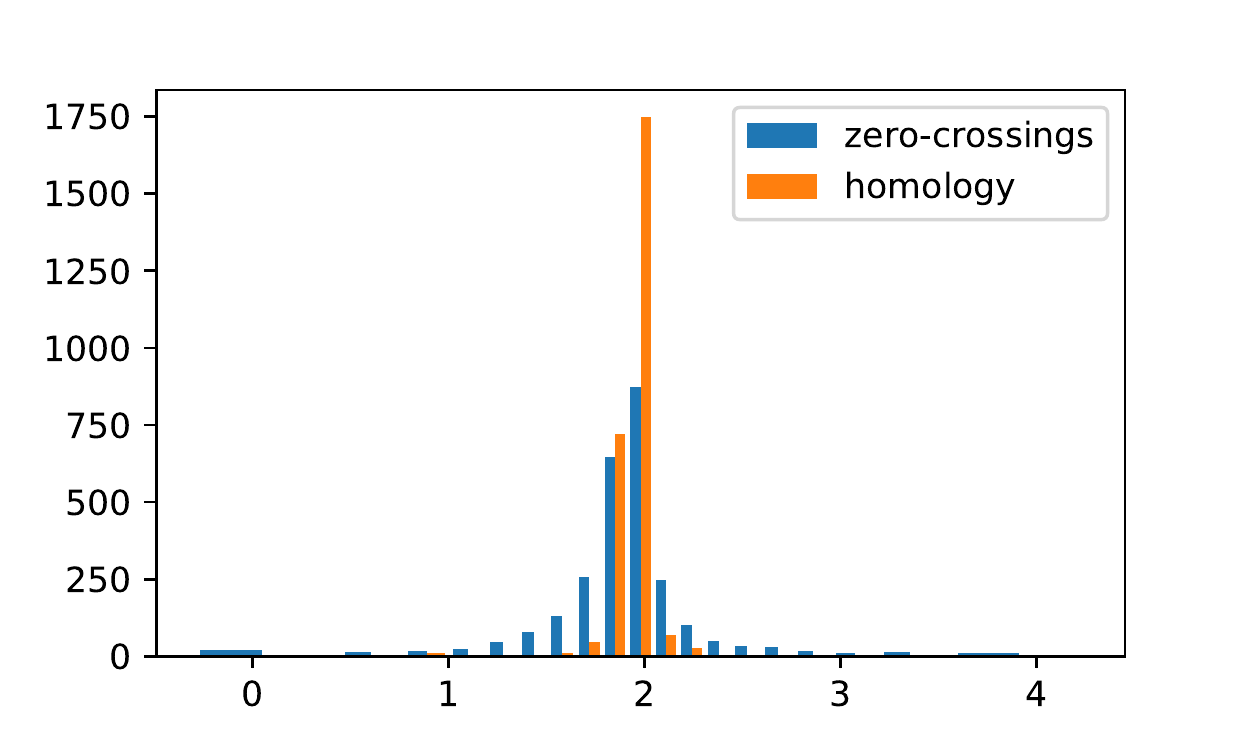}
	\caption{Histogram of the distances $(d_n)_n$, obtained for different odometric sequences.}
	\label{fig:dn_histogram}
\end{figure}

\begin{table}
	\centering
	\begin{tabular}{c|c c}
		& \multicolumn{2}{c}{Method} \\
		& zero--crossings & topological\\
		\hline
		TS& 490 & 46\\
		TL& 331 & 51\\
		CR& 0.69 & 0.96\\
		$R((d_n)_n)$& 0.23 & 0.05\\
	\end{tabular}
	\caption{Qualitative measures to compare the odometric sequences produces by the zero--crossings and the topological method.}
	\label{tbl:real_data_experiments}
\end{table}
After inspection of the first results, we noticed that the zero--crossings sequence contained many erroneous points in the time intervals when the car was at rest and that this was not the case for the topological method. We eliminated those time intervals from the results. Figure~\ref{fig:dn_histogram} shows that the sequence produced by the topological method is more concentrated around $C$, while the one from zero--crossings still contains a few elements near 0. This is confirmed by data in Table~\ref{tbl:real_data_experiments} - the sequence from the topological method leads to fewer elements outside of the region $[0.9C, 1.1C]$ and is on average much closer to the true circumference.

\begin{figure}
	\label{fig:velocities_comparison}
	\centering
	\includegraphics[width=0.58\textwidth]{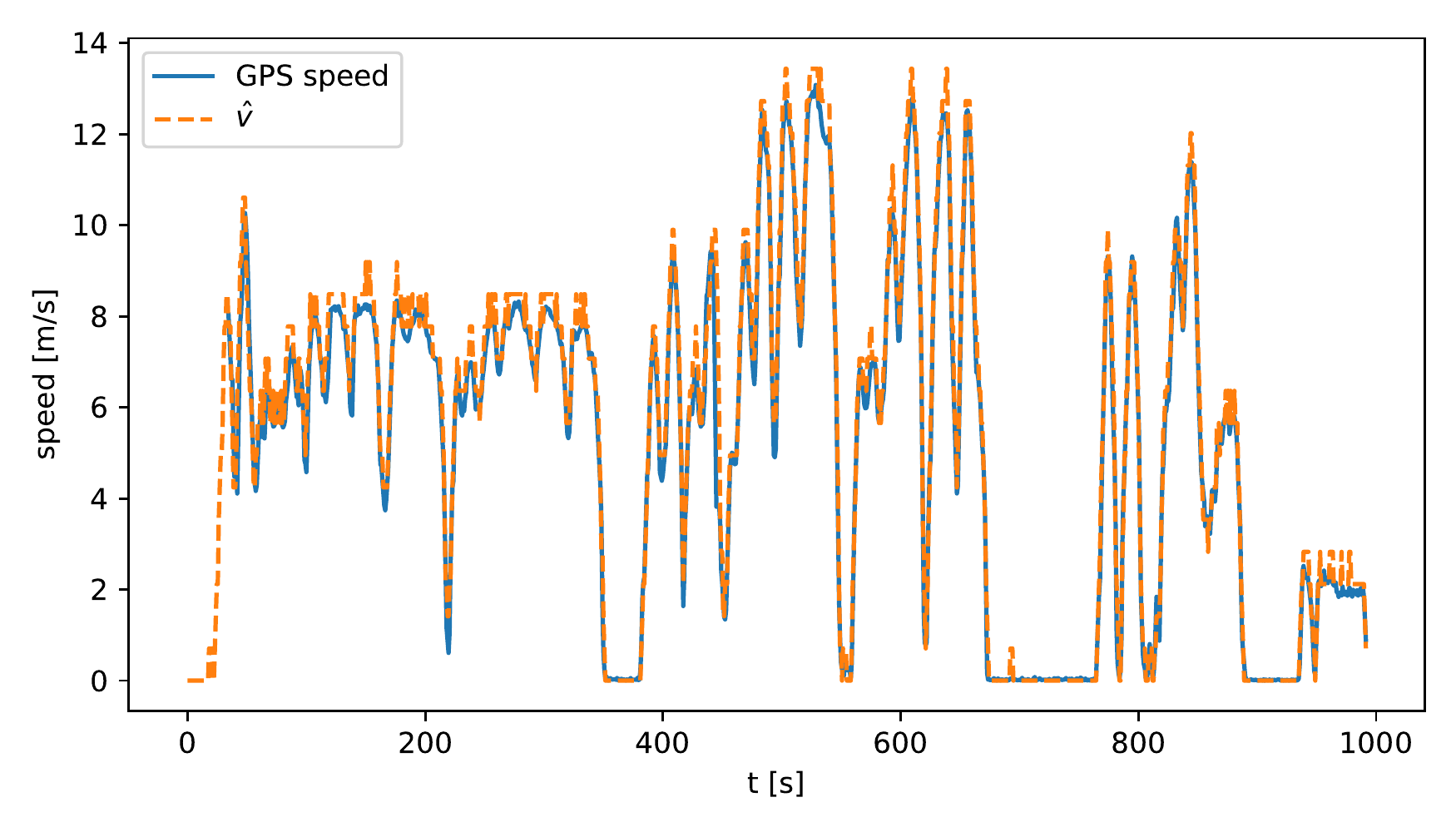}
	\includegraphics[width=0.41\textwidth]{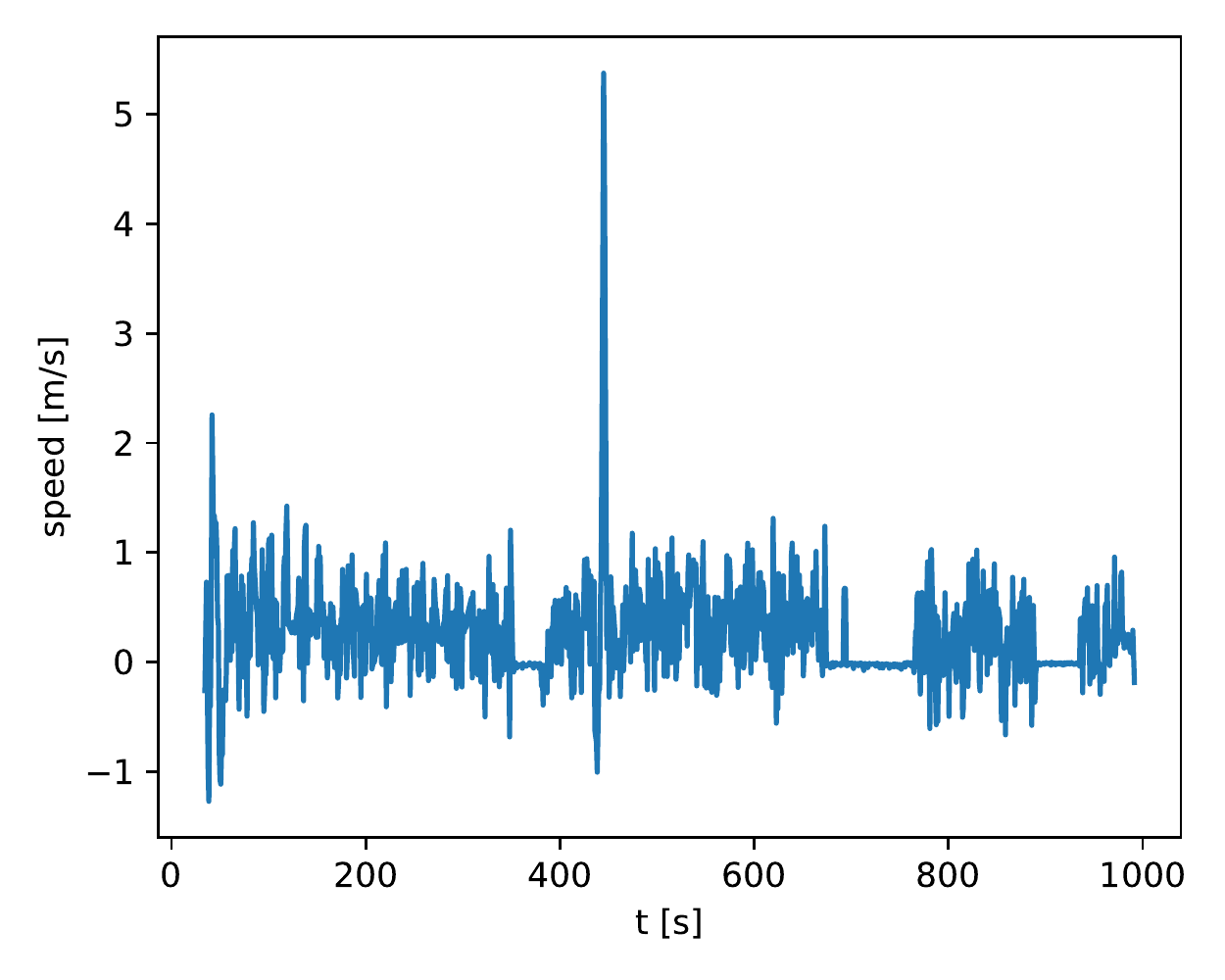}
	\caption{On the left, the speed estimation $\hat{v}$ compared with the GPS reference. On the right, the difference between the two.}
\end{figure}
Based on the odometric sequence $(t_n)_{n=1}^N$, we define an estimator of the displacement
\begin{equation*}
\hat{d}(t) = 1.98\sum_{n=1}^N \indicator_{t_n\leq t},
\end{equation*}
and an estimator of the speed $\hat{v}$, as the average speed over a 2.8--second window, estimated from $\hat{d}$. We picture $\hat{v}$ and the reference GPS speed in Figure~\ref{fig:velocities_comparison}. The difference between the two is comprised mostly between -0.5 and 1.1 m/s. Globally, the dynamics and the stops are correctly detected. The error could be optimized further by carefully selecting the length of the window.

\section{Conclusion and perspectives}
We formulated two inverse problems: estimation of $N$ and finding odometric sequences.
We establish a homogeneity property of persistence that we then use to propose an estimator of $N$.  We prove that the estimation is correct for a wide class of functions, even when the signal is obfuscated with additive noise. We propose a more robust variant $\hat{N}_c$ and a parameter--free version $\hat{N}_c^T$. We show that for certain types of functions that we describe, this estimator perform similarly to a standard approach. Using the estimation of $N$, we propose a method of obtaining odometric sequences. We show that the sequences are stable with respect to perturbations of the input signal. We demonstrate that they can be used for magnetic odometry in practice

In practical scenarios, we never observe an exact number of revolutions, so our hypothesis that $N\in\N$ is a significant limitation of our work.
It is not a noticeable problem in the case of the numerical experiments, because that constraint is satisfied for the synthetic examples, and for magnetic odometry, the signal has a single pair of local extrema per period.

Another limitation is the use of the greatest common divisor in~\eqref{eq:N_hat}. It is not robust to any perturbation (change in a single value in the set) and therefore, the whole procedure is very sensitive to outliers. An alternative would be to use solutions to the approximate common divisor (ACD) problem, like the Simultaneous Diophantine approximation (SDA)~\cite{galbraith_algorithms_2016}. They are shown to be stable with respect to perturbations of the individual elements, but require additional information. For example, before applying SDA in our case, we would need to give a sharp integer lower--bound for $\log_2(N)$. In addition, solutions to the ACD problem increase the computational cost of the estimation, preventing us from systematically investigating the potential gains. Nevertheless, based on a few synthetic examples, we believe that this avenue might dramatically improve our methods' performance in selected applications.

Finally, the presented method is designed for univariate data. While it can be applied to a single covariate from a multi--variate signal, a well--designed approach for several variables could make the estimation more robust. The restriction to univariate data is related to how we construct the persistence diagram and it is not immediately clear how to propose a signature for multi--variate data which shares similar properties as stated in Proposition~\ref{prop:periodic_multiplicity}.

\section*{Acknowledgments}
WR thanks the ANR TopAI chair (ANR--19--CHIA--0001) for financial support.
The authors thank Sysnav providing the experimental equipment and the help with the data acquisition.

\bibliography{references.bib}

\appendix
\section{Proof of Proposition~\ref{prop:periodic_multiplicity}}
\label{app:proof_periodic_multiplicity}

Notice that $\bar{f}_N = \bar{f}\circ\pi$, where $\pi:\Se^1\rightarrow\Se^1$ is the $N$-cover of $\Se^1$ defined by $z\mapsto z^N$. In the proof, we characterize the persistence module $(H_0\left((\bar{f}\circ\pi)^{-1}(\rbrack-\infty, \alpha\rbrack)\right)_\alpha, (\iota_{\alpha, \alpha'})_{\alpha\leq \alpha'})$.

Let $M=\max(\bar{f})$ and $k=\left\vert\ \bar{f}^{-1}(M)\ \right\vert$. Since $\bar{f}$ has a finite number of critical points, we can choose a small constant $\epsilon>0$ such that $\bar{f}^{-1}(\lbrack M-2\epsilon, M\lbrack)$ contains no critical points. The space $\bar{f}^{-1}(\rbrack-\infty, M-2\epsilon\rbrack)$ has $k$ connected components. Because $\pi$ is an $N$-cover, we have that $\vert(\bar{f}\circ\pi)^{-1}(M)\vert = Nk$, so that $(\bar{f}\circ\pi)^{-1}(\rbrack-\infty, M-2\epsilon\rbrack)$ has $Nk$ connected components. Therefore, for any $\alpha< M$,
$$H_0((\bar{f}\circ\pi)^{-1}(\rbrack-\infty, \alpha\rbrack)) = \oplus_{n=1}^N H_0(\bar{f}^{-1}(\rbrack-\infty, \alpha\rbrack)),$$
and the morphisms $(\iota_{\alpha, \alpha'})_{\alpha\leq \alpha'<M}$ are induced by direct sums. In particular, at the level of persistence diagrams, it implies that $D(\bar{f}\circ\pi)\cap A = \sqcup_{n=1}^N (D(\bar{f})\cap A)$, for any measurable set $A\subset \Omega_1 = \rbrack-\infty, M\lbrack\times \rbrack-\infty, M\lbrack$.

The circle $\Se^1 = (\bar{f}\circ\pi)^{-1}(\rbrack-\infty, M\rbrack)$ has a single connected, so that $\dim(\iota_{M-\epsilon}^{M})=Nk-1$. In addition, since the death of the essential component is set to $M$, at the level of persistence diagrams, we have $D(\bar{f}\circ\pi)(\Omega_2) = (Nk-1) + 1 = NK$, for $\Omega_2=\lbrack-\infty, M\lbrack\times [M, M]$ and $D(\bar{f}),\,D(\bar{f}\circ\pi)\subset \Omega_1\cup\Omega_2$. Since the local minimal values of $\bar{f}$ and $\bar{f}\circ\pi$ are the same,
\begin{align*}
D(\bar{f}\circ\pi)\cap A &= D(\bar{f}\circ\pi) \cap (A\cap\Omega_1) + D(\bar{f}\circ\pi) \cap (A\cap\Omega_2)\\
&= \sqcup_{n=1}^N D(\bar{f})\cap (A\cap\Omega_1) + \sqcup_{n=1}^N D(\bar{f})\cap(A\cap\Omega_2)\\
&= ND(\bar{f})\cap A.
\end{align*}
The same can be expressed in terms of rectangle measures~\cite{chazal_structure_2016}.

\section{Proof of Proposition~\ref{prop:assumption_is_necessary}}
Before showing Proposition~\ref{prop:assumption_is_necessary}, we introduce and prove two lemma. Lemma~\ref{lem:diagram_is_realizable} shows that a persistence diagram is realizable as the diagram of a sublevel set filtration of a space.

\begin{lemma}
	\label{lem:diagram_of_a_fct}
	Let $f:X\rightarrow \R$ be a continuous function on a compact domain $X$. Then,
	\begin{enumerate}
		\item $(\min(f), \max(f))\in D(f)$, and \label{lem:diagram_of_a_fct_1}
		\item $\supp(D(f))\subset \{(x,y)\mid \min(f)\leq x,\, y\leq \max(f)\}$. \label{lem:diagram_of_a_fct_2}
	\end{enumerate}
\end{lemma}
\begin{proof}
	We prove both statements by analyzing Algorithm~\ref{alg:persistence}.
	
	\emph{\ref{lem:diagram_of_a_fct_2}.} The first (resp. second) coordinates of points in the persistence diagram are local minimum (resp. maximum) values, so that they are greater (resp. smaller) than $\min(f)$ (resp. $\max(f)$).
	
	\emph{\ref{lem:diagram_of_a_fct_1}.} We claim that, when the \emph{for} loop terminates, there is a point $x\in D$ which satisfies two properties:
	\begin{itemize}
		\item the first coordinate of $x$ is $\min(f)$,
		\item the second coordinate is not set.
	\end{itemize}
	In the next line, we set the second coordinate of $x$ to $\max(f)$. To finish the proof, we show that such $x$ exists.
	Let $c$ be a global minimum of $s$. When we enter the \emph{for} loop, we add to $D$ a point $p$ with $\min(f)$ as the first coordinate. consider the iterations, let $v$ be a local maximum such that $c_1=c$ or $c_2=c$. Suppose, w.l.o.g that $c_1=c$. Then, $s(c)\leq f(c_2)$. If $f(c)<f(c_2)$, then $i=2$ and the second coordinate of $p$ is not set. If $\min(f) = f(c)=f(c_2)$, the second coordinate of one of the corresponding points is set to $f(v)$: it does not matter which as the other will satisfy both required conditions.
\end{proof}

Lemma~\ref{lem:diagram_is_realizable} is similar to \cite[Proposition 5.8]{lesnick_theory_2015}, but we include the proof for two reasons: first, to be self--contained and, second, to provide a simple construction specific to our case, where we characterize the period of $f$.
\begin{lemma}
	\label{lem:diagram_is_realizable}
	Let $D$ be a non--empty persistence diagram with a unique, most persistent point $p_0=(b_0, d_0)$. Assume that $D\setminus \Delta$ is finite and $D\subset [b_0,d_0]\times [b_0, d_0]$. Then, there exists a continuous function $f:\R\rightarrow \R$ with period $1$, such that $D=D(\restr{[0,1]})$, $f(0) = f(1)$.
\end{lemma}
\begin{proof}
	We first construct a candidate function $f$ as shown in Figure~\ref{fig:f_K_construction} and then, prove that it satisfies the desired properties.
	Let us enumerate the points in $D\setminus \Delta$, in lexicographic order of increasing birth and decreasing death values: $(b_0,d_0),\ldots (b_{K-1}, d_{K-1})$, where $K=\vert D\setminus \Delta\vert$. Then, consider the function
	\begin{equation*}
	\begin{array}{cccc}
	f_K:&\{0,\ldots,2K\}&\rightarrow&\R \\
	& k&\mapsto&
	\begin{cases*}
	d_{k/2}, \text{ if $k$ is pair}\\
	b_{(k-1)/2}, \text{otherwise,}
	\end{cases*}
	\end{array}
	\end{equation*}
	with the convention $d_K=d_0$, and the linear interpolation $f:[0,1]\rightarrow\R$.
	\begin{figure}
		\centering
		\includegraphics[width=0.99\textwidth]{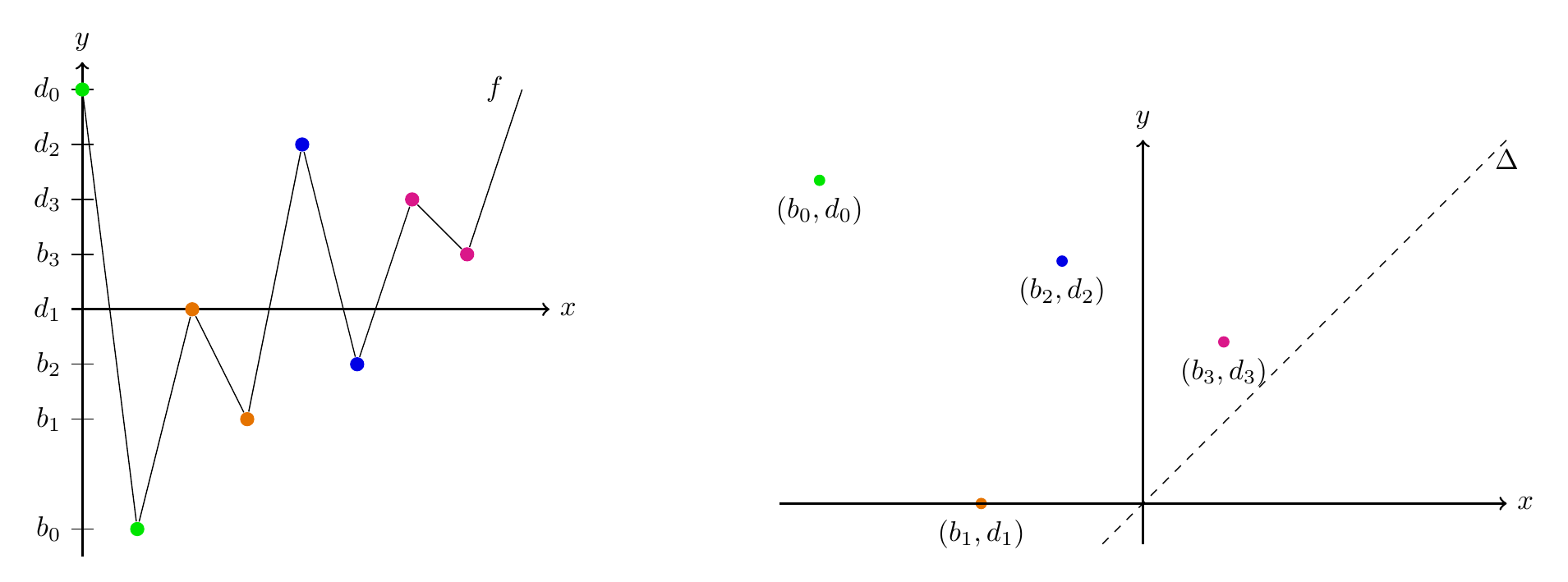}
		\caption{Illustration of the function $f$.}
		\label{fig:f_K_construction}
	\end{figure}
	
	First, note that $f$ is continuous and that $f(0)=f_K(0)=d_0=d_K=f_K(2K)=f(1)$.
	The sequence of extrema is aperiodic if either $K=0$ or $(b_0, d_0)\neq (b_K, d_K)$. An aperiodic sequence leads to $f_K$ aperiodic and so $f$ can be extended to $\R$ to have period 1. If $(b_0, d_0) = (b_K, d_K)$, modify the subdivision of $[0,1]$ to be non--uniform, in which case the period of $f$ is also 1.
	
	We argue that $D=D(f)$. It can be seen by noticing that every local minimum is paired with the local maximum on its left. Formally, we calculate $D(f)$ with Algorithm~\ref{alg:persistence}. The local minima and maxima of $f$ are $(x_{b_k})_{k=1}^K,\, (x_{d_k})_{k=1}^K$ respectively. We only need to consider, for every $d_k$,
	what point it is the second coordinate of. The connected component $I$ of $x_{d_k}$ in $f^{-1}(\rbrack -\infty, d_k\rbrack)$ is partitioned into two non--empty sets $I_- = \lbrack 0, x_{d_k}\rbrack \cap I$ and $I_+ = \lbrack x_{d_k}, 1\rbrack\cap I$. Denote $M\coloneqq \max(f)$ and suppose that $d_k<M$. Then, since the values of the local minima are increasing by definition of $f$, $\min \restr{f}{I_-}<\min\restr{f}{I_+}$ and $x_{b_k}=\argmin \restr{f}{I_+}$, so $d_k$ becomes the second coordinate of $b_k$. Therefore, there are exactly $N$ copies of each point $(b_k,d_k)$, such that $d_k<M$. If $d_k=M$, the second coordinate of a point whose second coordinate has not been set yet, is set to $d_k$.
\end{proof}

\begin{proof}[Proof of Proposition~\ref{prop:assumption_is_necessary}]
	Let $\mu = \frac{1}{N(\restr{f}{[0,1]})}\mu_{\restr{f}{[0,1]}}$. The measure $\mu$ is a persistence diagram: by definition, it is a finite sum of weighted point measures. The weights are integral, because $N(\restr{f}{[0,1]})$ divides the multiplicity of every point.
	
	Thanks to Lemma~\ref{lem:diagram_of_a_fct}, we can apply Lemma~\ref{lem:diagram_is_realizable}to $\mu_{\restr{f}{[0,1]}}$, which gives us $g:\R\rightarrow\R$ such that $\mu_{\restr{g}{[0,1]}} = \mu = \frac{1}{N(\restr{f}{[0,1]})}\mu_{\restr{f}{[0,1]}}$.
\end{proof}

\section{$\hat{N}(S)$ for a signal with white noise}
\label{app:sampled_signal}
In Section~\ref{sec:method_multiplicity_points}, we discuss the estimation of $\gamma(1)$ in the presence of noise. The aim of this section is to derive guarantees in the case where the noise entries are independent. Such noise is called "white noise" and it is often considered in applications. Continuous equivalents of white noise processes exist, but they are difficult to analyze and many objects, like persistence diagrams, are not defined for such processes. Since this consideration is motivated by practical considerations, we choose to derive guarantees on a sample from the signal, sampled at $\freq>0$. For example, in magnetic odometry motivating this work, $\freq= 125Hz$.

Let $M = \lfloor\freq\rfloor$. We define the sampling operator $L:C^0([0,1])\rightarrow \R^{M}$ and the piecewise linear interpolation operator $F:\R^{M+1}\rightarrow C^0([0,\tfrac{M}{\freq}])$. Specifically, for $h\in C^0[0,1]$ and $a\in\R^{M+1}$,
\begin{align*}
L(h) &\coloneqq (h(\tfrac{m}{\freq}))_{m=0}^M,\qquad\\
F(a)(t)&\coloneqq a_{m}(\freq t - (m-1)) + a_{m-1}(m - \freq t),\ \text{ for any } t\in \lbrack\tfrac{m-1}{\freq}, \tfrac{m}{\freq}\lbrack.
\end{align*}
We denote the composition of the two by $T=F\circ L: \mathcal{C}^3([0,1]) \rightarrow \mathcal{C}^0([0,\tfrac{M}{\freq}])$.
We examine the approximation obtained by interpolating a uniform sample from $f\circ\gamma$.
\begin{lemma}
	\label{lemma:sampling_stability}
	Suppose that $f$ and $\gamma$ are $\mathcal{C}^3$. Then, for $t\in\left\lbrack 0,\frac{M}{\freq}\right\lbrack$,
	$$\vert T(f\circ\gamma)(t) - (f\circ\gamma)(t)\vert \leq \tfrac{2}{\freq^2}(\Vert f''\Vert_\infty \Vert \gamma'\Vert_\infty^2 + \Vert f'\Vert_\infty \Vert \gamma''\Vert_\infty) + \mathcal{O}(\freq^{-3}).$$
\end{lemma}
\begin{proof}
	First, we show that
	$\vert T(h) - h \vert = \mathcal{O}\left(\tfrac{\Vert h''\Vert_\infty}{\freq^2}\right).$ For $m$ such that $t\in\left\lbrack \tfrac{m-1}{\freq},\tfrac{m}{\freq}\right\lbrack$,
	\begin{align*}
	\vert T(h)(t) - h(t)\vert
	=& \left\vert (h\left(\tfrac{m}{\freq}\right) - h(t))(\freq t - (m-1)) + \left(h\left(\tfrac{m-1}{\freq}\right) - h(t)\right)(m - \freq t)\right\vert\\
	=& \vert (h'(\tfrac{m}{\freq})(t-\tfrac{m}{\freq}) + \tfrac{1}{2}h''(t_m^*)\left(t-\tfrac{m}{\freq}\right)^2)(\freq t - (m-1)) \\
	&+ (h'\left(\tfrac{m-1}{\freq}\right)\left(t-\tfrac{m-1}{\freq}\right) + \tfrac{1}{2}h''(t_{m-1}^*)\left(t-\tfrac{m-1}{\freq}\right)^2)(\freq t - (m-1))\vert,
	\end{align*}
	where $t_{k-1}^*\in [\tfrac{m-1}{\freq}, t],\ t_k^*\in [t,\tfrac{m}{\freq}]$ are given by the Taylor--Lagrange expansion of $h$. By an expansion of $h'$ around $\tfrac{m-1}{\freq}$, 
	\begin{align*}
	\vert T(h)(t) - h(t)\vert \leq& \tfrac{1}{\freq^2}\left(\left\vert h''\left(\tfrac{m-1}{\freq}\right)\right\vert + \vert h''(t_1)\vert + \tfrac{1}{2\freq}h'''(t_2)\right) \\
	\leq& \tfrac{1}{\freq^2}\left(2\Vert h''\Vert_\infty + \mathcal{O}\left(\tfrac{1}{\freq}\right)\right).
	\end{align*}
	We conclude by applying the above to $h=f\circ\gamma$.
\end{proof}
The following proposition is an application of the stability result to $[0,\tfrac{M}{\freq}]$.
\begin{proposition}
	\label{prop:sampled_diagram_distance}
	Let $f,\,\gamma$ be $\mathcal{C}^3$ as above. In addition, suppose that $\freq$ is large enough so that $[0,\tfrac{M}{\freq}]$ contains all the local extrema of $f\circ\gamma$. Then,
	$$ \bottleneck(D(T(f\circ\gamma)), D(f\circ\gamma)) \leq \tfrac{2}{\freq^2}(\Vert f''\Vert_\infty \Vert \gamma'\Vert_\infty^2 + \Vert f'\Vert_\infty \Vert \gamma''\Vert_\infty) + \mathcal{O}(\freq^{-3}).$$
\end{proposition}
\begin{proof}
	Since all extrema of $f\circ\gamma$ are included in $I = [0,\tfrac{M}{\freq}]$, we have equality between the diagrams $D\left(\restr{f\circ\gamma}{I}\right) = D(f\circ\gamma)$. Using the stability of the persistence diagram, on the filtrations induced by $T_\freq(f\circ\gamma)$ and $f\circ\gamma$ on I, we conclude with the Lemma~\ref{lemma:sampling_stability}.
\end{proof}

Proposition~\ref{prop:sampled_diagram_distance} relates the diagram of the noise--less signal to the sampled version. Let us now consider a noisy sample from $p\circ\gamma$. More precisely,
$$(S_m)_{m=1}^M = (f(\gamma(\tfrac{m-1}{\freq})) + W_m)_{m=1}^M,$$
where $W_m$ are independent, identically distributed centered Gaussian random variables with standard deviation $\sigma$.
\begin{proposition}
	\label{prop:iid_gaussian_N_stability}
	Let $w$ and $0<\tau<\delta/3$ be such that $\alpha \coloneqq \tau/2 - \tfrac{1}{\freq^2}C_{f,\gamma}>0$, where
	\begin{align*}
		C_{f,\gamma}=&\Vert f''\Vert_\infty \Vert \gamma'\Vert_\infty^2 + \Vert f'\Vert_\infty \Vert \gamma''\Vert_\infty\\
		&+ \tfrac{1}{2}(\Vert f'''\Vert_\infty\Vert \gamma'\Vert_\infty^3 + 3\Vert f''\Vert_\infty \Vert \gamma'\Vert_\infty \Vert \gamma''\Vert_\infty + \Vert f'\Vert_\infty \Vert \gamma'''\Vert_\infty).
	\end{align*}
	Then,
	$$P(\hat{N}(F((S_m)_{m=1}^M))=N, \tau) \geq (1-\phi(\tfrac{\alpha}{\sigma}))^\freq.$$
\end{proposition}
\begin{proof}
	We show that the linear approximation of the exact signal is good enough, so that the entries can be perturbed with noise. Indeed, $\Vert F((S_m)_{m=1}^M) - f\circ\gamma\Vert_\infty \leq \Vert (W_m)\Vert_\infty + \Vert T(f\circ\gamma) - f\circ\gamma\Vert_\infty.$ By Lemma~\ref{lemma:sampling_stability},
	$\Vert F((S_m)_{m=1}^M) - f\circ\gamma\Vert_\infty \leq \Vert (W_m)\Vert_\infty + \tfrac{1}{\freq^2}C_{f,\gamma}.$ Since $W_m$ are independent, centered normal variables with variance $\sigma^2$, $P(W_m\leq \alpha) = (1-\phi(\tfrac{\alpha}{\sigma}))^\freq$. For $\alpha = \tau/2 - \tfrac{1}{\freq^2}C_{f,\gamma}$, $\Vert F((S_m)_{m=1}^M) - f\circ\gamma\Vert_\infty \leq \alpha$ with probability $(1-\phi(\tfrac{\alpha}{\sigma}))^\freq$. By Proposition~\ref{prop:N_stability},
	$P(\hat{N}(F((S_m)_{m=1}^M))=N, \tau)<(1-\phi(\tfrac{\alpha}{\sigma}))^\freq$.
\end{proof}

The bound in Proposition~\ref{prop:iid_gaussian_N_stability} is decreasing in $\freq$, as soon as the approximation is good enough: $\alpha>0$. There is a trade--off between the probability of $\hat{N}$ being correct and the stability of the precision of the odometric sequence, which increases with $w$.
\section{Proof of Lemma~\ref{lemma:min_max_min}}
\label{app:min_max_min_lemma}
In the proof, we propose to use the explicit construction of the persistence diagram of the sublevel sets of a function. One such construction is described in Algorithm~\ref{alg:persistence}.
The image $A = f([x_1, x_2])$ is an interval. Suppose that $f(x_1)< f(x_2)$. We will prove that there exists a local maximum $c\in [x_1,x_2]$, such that $f(x_2) + 2\delta\leq f(c)$.
Since $x_2$ is a local minimum, by Algorithm~\ref{alg:persistence}, $(f(x_2),d)\in D(f)$, for some $d\in\R$. Let $x_d$ be the local maximum which set the second coordinate of $(f(x_2), d)$ and $I$ be the connected component of $x_d$ in $f^{-1}(\rbrack -\infty, d\rbrack)$. Note that by the separation property, $d\geq f(x_2)+2\delta$.
Now, we examine two cases. First, let us suppose $x_d<x_2$. Then, $x_2=\argmin_{x>x_d, x\in I}f(x)$ by line 8 in Algorithm~\ref{alg:persistence} and $f(x_1)< f(x_2)$ by line 9 imply that $x_1<x_d<x_2$. We set $c=x_d$.
If $x_d>x_2$, then, $x_2=\argmin_{x>x_d, x\in I}f(x)$. Recall that $I$ is the connected component containing $x_d$ in $f^{-1}(\rbrack-\infty, f(x_d)\rbrack)$. By assumption $f(x_1)<f(x_2)$, we have that $x_1\notin I$, so there is $c\in[x_1,x_2]$ such that $c\notin f^{-1}(\rbrack -\infty,f(x_d)\rbrack)$ and therefore, $f(c)\geq f(x_d)\geq f(x_2)+2\delta$. 
The case $f(x_1)>f(x_2)$ is treated in the same way. Note that $f(x_1)\neq f(x_2)$, since we use the total lexicographic order $(f(x_1), x_1), (f(x_2), x_2)$.

\end{document}